\documentclass[11pt]{article} 

\usepackage[utf8]{inputenc} 

\usepackage{geometry} 
\usepackage{graphicx} 
\usepackage{color}

\definecolor{Red}{rgb}{1,0,0}
\definecolor{Blue}{rgb}{0,0,1}
\definecolor{Olive}{rgb}{0.41,0.55,0.13}
\definecolor{Green}{rgb}{0,1,0}
\definecolor{MGreen}{rgb}{0,0.8,0}
\definecolor{DGreen}{rgb}{0,0.55,0}
\definecolor{Yellow}{rgb}{1,1,0}
\definecolor{Cyan}{rgb}{0,1,1}
\definecolor{Magenta}{rgb}{1,0,1}
\definecolor{Orange}{rgb}{1,.5,0}
\definecolor{Violet}{rgb}{.5,0,.5}
\definecolor{Purple}{rgb}{.75,0,.25}
\definecolor{Brown}{rgb}{.75,.5,.25}
\definecolor{Grey}{rgb}{.5,.5,.5}

\usepackage{booktabs} 
\usepackage{array} 
\usepackage{paralist} 
\usepackage{verbatim} 
\usepackage{subfigure} 
\usepackage{amsmath,amssymb,amsthm}
\usepackage{cancel}

\usepackage{fancyhdr} 
\usepackage{sectsty}
\allsectionsfont{\sffamily\mdseries\upshape} 

\theoremstyle{plain}

\newtheorem{theorem}{Theorem}[section] 
 
\newtheorem{corollary}{Corollary}[section]
\newtheorem{claim}{Claim}[section]
\newtheorem{lemma}{Lemma}[section]
\newtheorem{conj}{Conjecture}

\theoremstyle{remark}

\newtheorem{remark}{Remark}

\theoremstyle{definition}

\title{The Entropy Power Inequality and Mrs. Gerber's Lemma for Abelian Groups of Order ${2^n}$}
\author{ Varun  Jog,  Venkat Anantharam\\
Department of Electrical Engineering and Computer Sciences \\ University of California, Berkeley\\
\texttt{\small jogvarun@gmail.com,ananth@eecs.berkeley.edu}}

\begin{document}
\maketitle
\allowdisplaybreaks

\begin{abstract}

Shannon's Entropy Power Inequality can be viewed as characterizing the minimum differential entropy achievable by the sum of two independent 
random variables with fixed differential entropies. 
The entropy power inequality has played a key role in resolving a number of problems in information theory. It is therefore interesting to examine the existence of a similar inequality for discrete random variables.
In this paper we obtain an entropy power inequality for random variables taking values in an abelian group of order $2^n$, i.e. for such a group $G$ we explicitly characterize the function $f_G(x,y)$ giving
the minimum entropy of the sum of two independent $G$-valued random variables with respective entropies $x$ and $y$. Random variables achieving the extremum in this inequality 
are thus the analogs of Gaussians in this case, and these are also determined. It turns out that $f_G(x,y)$ is convex in $x$ for fixed $y$ and, by symmetry, convex in $y$ for fixed $x$. This is a generalization to abelian 
groups of order $2^n$ of the result known as Mrs. Gerber's Lemma.\\ 

\end{abstract}

\section{Introduction}
The Entropy Power Inequality (EPI) relates to the so called ``entropy power" of $\mathbb{R}^n$-valued random variables having densities with well defined differential entropies. It was first proposed by Shannon in $1948$ \cite{shannon1948},  who also gave sufficient conditions for equality to hold. The entropy power of 
an $\mathbb{R}^n$-valued random variable $\mathbf{X}$ is defined as the per-coordinate variance of a circularly symmetric $\mathbb{R}^n$-valued Gaussian random variable with the same differential entropy as $\mathbf{X}$.
\\
\begin{theorem}[Entropy Power Inequality]
For an $\mathbb{R}^n$-valued random variable $\mathbf{X}$, the entropy power of $\mathbf{X}$ is defined to be
\begin{equation}\label{ep}
N(\mathbf{X}) = \frac{1}{2\pi e}e^{\frac{2}{n}h(\mathbf{X})},
\end{equation}
where $h(\mathbf{X})$ stands for the differential entropy of $X$. 
Now let $\mathbf{X}$ and $\mathbf{Y}$ be independent $\mathbb{R}^n$-valued random variables. The EPI states that entropy power is  a super-additive function, that is
\begin{equation}\label{epi}
N(\mathbf{X}) + N(\mathbf{Y}) \leq N(\mathbf{X+Y}),
\end{equation}
with equality if and only if $\mathbf{X}$ and $\mathbf{Y}$ are Gaussian with proportional covariance matrices.
\end{theorem}
Shannon used a variational argument to show that $\mathbf{X}$ and $\mathbf{Y}$ being Gaussian with proportional covariance matrices and having the required entropies is a stationary point for $h(\mathbf{X+Y})$, but this did not exclude the possibility of it being a local minimum or a saddle point. The first rigorous proof of (\ref{epi}) was given by Stam \cite{stam1959} in $1959$ based on an identity
communicated to him by N. G. De Bruijn, which couples Fisher information with  differential entropy.  Stam's proof was further simplified by Blachman \cite{blachman1965}. Lieb \cite{lieb1978} gave a proof of the EPI using a strengthened Young's inequality. More recently, Verd{\'u} and Guo \cite{verdu2006} gave a  proof without invoking Fisher information, by using the relationship between mutual information and minimum mean square error (MMSE) for Gaussian channels. Rioul \cite{rioul2011} managed to give a proof sidestepping Fisher information as well as MMSE estimates.
\smallskip

The EPI has a played a key role in the solution of a number of communication problems. It is generally used to prove converses of coding theorems when Fano's inequality is insufficient to prove optimality. Some famous examples consist of Bergmans's solution to the Gaussian broadcast channel problem \cite{bergmans1973}, Leung-Yan-Cheong and Hellman's determination of  the secrecy capacity of a Gaussian wire-tap channel \cite{leung1978}, Ozarow's solution to the scalar Gaussian source two-description problem \cite{ozarow1980}, Oohama's solution to the quadratic Gaussian CEO problem \cite{oohama1998}, and recently Weingarten, Steinberg and Shamai's solution to the multiple-input multiple-output Gaussian broadcast channel problem \cite{weingarten2006mimo}.
\smallskip

The EPI has been generalized in a number of ways. Costa \cite{costa1985} strengthened the inequality when one of random variables was Gaussian. In particular, Costa showed that if independent Gaussian noise is added to an arbitrary multivariate random variable, the entropy power of the resulting random variable is concave in the variance of the added noise. Dembo \cite{dembo1989} reduced Costa's inequality to an equivalent inequality in terms of Fisher information and proved this inequality. Vilani \cite{villani2000} further simplified Dembo's proof. Zamir and Feder \cite{zamir1993generalization} generalized the scalar EPI using linear transformations of random variables. T. Liu and Viswanath \cite{viswanathliu2007} obtained a generalization of the EPI by considering a covariance-constrained optimization problem, motivated by the problems of the capacity region of the vector Gaussian broadcast channel and of distributed source coding with a single quadratic distortion constraint. R. Liu, T. Liu, Poor and Shamai \cite{liu2010vector} gave a vector generalization of Costa's EPI. The EPI for general independent random variables and the corresponding Fisher information inequalities have also been used to prove strong versions of the central limit theorm, with convergence in relative entropy. Artstein, Ball, Barthe, and Naor \cite{artstein2004} showed that the non-Gaussianness (divergence with respect to a Gaussian random variable with identical first and second moments) of the sum of independent and identically distributed random variables is monotonically non-increasing.. Simplified proofs of this result were later given Tulino and Verd\'{u} \cite{tulino2006} and by Madiman and Barron \cite{madiman2007generalized}. 
\smallskip

There have also been several attempts to obtain discrete versions of the EPI. For the binary symmetric channel (BSC), Wyner and Ziv \cite{wyner1973part1},\cite{wyner1973part2} proved a result called Mrs. Gerber's Lemma (MGL), see Theorem \ref{mgl} below, which was extended to arbitrary binary input-output channels by Witsenhausen \cite{witsenhausen1974}. Shamai and Wyner \cite{shamai1990} used MGL to give a binary analog of the EPI. Harremo\"{e}s and Vignat \cite{harremoes2003epi} proved a version of the EPI for binomial random variables with parameter $\frac{1}{2}$. Sharma, Das and Muthukrishnan \cite{sharma2011epi} expanded the class of binomial random variables for which Harremo\"{e}s's EPI holds. Johnson and Yu \cite{johnson2010monotonicity} gave a version of the EPI for discrete random variables using the notion of Renyi thinning. 
\smallskip

In this paper we take a different approach towards getting a discrete analog of the EPI. Notice that even though the EPI is interpreted as an inequality in terms of the ``entropy power" of random variables, it is essentially a sharp lower bound on the differential entropy of a sum of independent random variables in terms of their individual differential entropies. If we are dealing with discrete random variables, as long the ``sum" operation is defined we can arrive at an analogous lower bound, except with entropies instead of differential entropies. A natural case to consider is when the random variables take values an abelian group $G$ and to define the function $f_G: [0, \log |G|] \times [0, \log |G|] \to [0, \log |G|] $ by 
\begin{equation}		\label{fgdefn}
f_G(x,y) = \min_{H(X)=x, H(Y)=y} H(X+Y).
\end{equation} 
We can then exploit the group structure and try to arrive at the explicit form of $f_G$.\\ 
 
A closely related function has been studied by Tao \cite{tao2010sumset} in which the sumset theory of Plunnecke and Ruzsa \cite{taovu2006} has been reinterpreted using  entropy as a proxy for the cardinality of a set. The sumset and inverse sumset inequalities in \cite{tao2010sumset} were further proved for differential entropy in \cite{kontoyiannis2012sumset}. 
\smallskip

Let us now consider two special cases:  $ G = \mathbb{Z}_2$ and $G = \mathbb{R}$. In the first case we note that on $\mathbb{Z}_2$, there is a unique distribution (up to rotation) corresponding to a fixed value of entropy. We can use this to simplify $f_{\mathbb{Z}_2}$ by writing it in terms of the inverse of binary entropy, $h^{-1}:[0,\log 2] \to [0, \frac{1}{2}]$
\begin{equation}
f_{\mathbb{Z}_2}(x,y) = h(h^{-1}(x) \star h^{-1}(y)).
\end{equation}
This is precisely the function for which Wyner and Ziv's MGL is applicable, in fact we can restate MGL in terms of $f_{\mathbb{Z}_2}$:
\begin{theorem}[Mrs. Gerber's Lemma] \label{mgl}
$f_{\mathbb{Z}_2}(x,y)$ is convex in $y$ for a fixed $x$, and by symmetry convex in $x$ for a fixed $y$.
\end{theorem}

For the case of $G = \mathbb{R}$ it is worthwhile to note that the function $f_\mathbb{R}$, which can be written explicitly as
\begin{equation}
f_\mathbb{R}(x,y) = \frac{1}{2} \log \left(e^{2x}+e^{2y} \right) 
\end{equation}
satisfies the convexity property described by MGL. In fact $f_\mathbb{R}$ is jointly convex in $(x,y)$. We can however easily check that $f_{\mathbb{Z}_2}$ is not jointly convex in $(x,y)$ since $f_{\mathbb{Z}_2}(x,x) > x = \frac{x}{\log 2} f_{\mathbb{Z}_2}(\log 2, \log 2) + \left(1 - \frac{x}{\log 2}\right) f_{\mathbb{Z}_2}(0,0)$. 

It seems natural to make the following conjecture:
\begin{conj}[Generalized MGL]\label{conjecture}
If $G$ is a finite abelian group, then $f_G(x,y)$ is convex in $x$ for a fixed $y$, and by symmetry convex in $y$ for a fixed $x$.
\end{conj}
Witsenhausen \cite{witsenhausen1974} and Ahlswede and K{\"o}rner \cite{ahlswede1974} attempted to generalize MGL by defining $g(x)$ to be   the minimum output entropy of a channel subject to a fixed input entropy $x$. They showed that $g(x)$ is convex for all binary input - binary output channels, but that counterexamples to this convexity exist for  other channels. They resolve this issue by providing a version of MGL based on the convex envelope of $g(x)$. Our function $f_G(x,y)$ can be thought of as related to the $g$ function in this line of work, but it differs in the key aspect that the `channel' is not fixed. To connect to this line of work, we can think of the capacity of the channel as being fixed (subject to it being an additive noise channel). We are then looking at the worst possible (in terms of minimum mutual information $I(X+Y;X)$) input and channel distributions, while fixing the input entropy and the channel capacity. 
\smallskip

We have carried out simulations to test Conjecture \ref{conjecture} for $\mathbb{Z}_3$ and $\mathbb{Z}_5$ and it appears to hold for these groups. In this paper we prove Conjecture \ref{conjecture} for all abelian groups $G$ of order $2^n$. In fact we arrive at an explicit description of $f_G$ in terms of $f_{\mathbb{Z}_2}$ for such groups. We also characterize those distributions where the minimum entropy is attained -- these distributions are in this sense analogous to Gaussians in the real case. Our results support the intuition that to minimize the entropy of the sum, the random variables $X$ and $Y$ should be supported on the smallest possible subgroup of $G$ (or cosets of the same) which can support them while satisfying the constraints $H(X) = x$ and $H(Y) = y$.
\smallskip

The structure of the document is as follows. In section $2$ we consider the function $f_{\mathbb{Z}_2}$  and derive certain lemmas regarding the behaviour of $f_{\mathbb{Z}_2}$ along lines passing through the origin. In section $3$, we use the preceding lemmas to explicitly compute $f_{\mathbb{Z}_4}$. This can be thought of as the induction step toward evaluating $f_{\mathbb{Z}_{2^n}}$. In section $4$ we use induction and determine the form of $f_{\mathbb{Z}_{2^n}}$. In section $5$ we show that if $G$ is abelian and of order $2^n$, then $f_G = f_{\mathbb{Z}_{2^n}}$.
Since $f_G$ is explicitly determined for all abelian groups of order $2^n$ we have in effect proved an EPI for such groups. Further, the $f_G$ we find verifies Conjecture \ref{conjecture} and so proves MGL for all abelian groups of order $2^n$. In section $6$ we provide some generalizations of our result that are likely to be of interest. Notably, we study the minimum entropy of a sum of $k\geq 2$ independent $G$-valued random variables of fixed entropies for $G$ of order $2^n$, and give an iterative expression to compute this minimum in terms of $f_G$.

\section{Preliminary Inequalities}

In this section we prove a few key lemmas which are needed to prove our EPI and MGL for $\mathbb{Z}_4$, then for $\mathbb{Z}_{2^n}$, and finally for abelian groups $G$ of order $2^n$. Consider $f:[0, \log 2] \times [0, \log 2] \to [0, \log 2]$ given by
$$f(x,y) = h(h^{-1}(x) \star h^{-1}(y))~.$$
Of course $f = f_{\mathbb{Z}_2}$, where $f_{\mathbb{Z}_2}$ is
defined in equation (\ref{fgdefn}), but it is convenient to drop the subscript in this section.

For our first lemma,  we consider lines of slope 
$0 \le \theta \le \infty$ passing through the origin. The result we wish to prove is:
\begin{lemma} \label{dfdx}
$\frac{\partial f}{\partial x}$  strictly decreases along lines through the origin having slope $\theta$, where $0<\theta<\infty$.
\end{lemma} 
\begin{remark}When $\theta = 0$, $\frac{\partial f}{\partial x} $ is constant and is equal to $1$ and when $\theta = +\infty$, $\frac{\partial f}{\partial x}$ is constant and equal to $0$. The above lemma claims that for all other values $\theta \in (0, \infty)$, $\frac{\partial f}{\partial x}(x, \theta x)$ strictly decreases in $x$.
\end{remark}
\begin{proof}
For the proof, refer to Appendix \ref{lemma1}.
\end{proof}
\begin{lemma}\label{concavity}
$f(x,y)$ is concave along lines through the origin. More precisely, $f(x,y)$ is concave along the line $y = \theta x$ when $0 \leq \theta \leq \infty$, and strictly concave along this line for $0 < \theta < \infty$.
\end{lemma}
\begin{proof}[Proof of Lemma \ref{concavity}]
When $\theta = 0$ or $\infty$, $f(x,y)$ is linear along the line $y = \theta x$, thus concave. For $0 < \theta < \infty$, by Lemma \ref{dfdx}, we have that $\frac{\partial f}{\partial x}$ strictly decreases along lines through the origin. By symmetry, it follows that $\frac{\partial f}{\partial y}$ also strictly decreases along lines through the origin. Since
\begin{equation}
\frac{df(x,\theta x)}{dx} =  \frac{\partial f}{\partial x} (x, \theta x) + \theta  \frac{\partial f}{\partial y}(x, \theta x)~,
\end{equation}
it is immediate that $\frac{df(x,\theta x)}{dx}$ also strictly decreases in $x$, which means that $f(x,y)$ is strictly concave along the line $y = \theta x$.
\end{proof}

\begin{lemma} \label{unique}
If  $(x_1,y_1), (x_2,y_2) \in (0,\log2)\times(0,\log2)$ and $(\frac{\partial f}{\partial x} ,\frac{\partial f}{\partial y})\bigg|_{(x_1,y_1)} = (\frac{\partial f}{\partial x} ,\frac{\partial f}{\partial y})\bigg|_{(x_2,y_2)}$ then $(x_1,y_1)=(x_2,y_2)$.
\end{lemma}
\begin{remark}
The above lemma says that in the interior of the unit square, the pair of partial derivatives at a point uniquely determine the point. That this fails on the boundary is seen from the fact that for any point of the form $(x,0)$ the pair of partial derivatives evaluates to $(1,0)$ and for every point of the form $(0,y)$ it is $(0,1)$.
\end{remark}
\begin{proof}[Proof of Lemma \ref{unique}]
Without loss of generality, assume $x_1 \leq x_2$. We consider two cases: $y_1 \geq y_2$ or $y_1 < y_2$.\\ Suppose $y_1 \geq y_2$, in this case we have
\begin{equation}
\frac{\partial f}{\partial x}\bigg|_{(x_1,y_1)} \leq \frac{\partial f}{\partial x}\bigg|_{(x_2,y_1)} \leq  \frac{\partial f}{\partial x}\bigg|_{(x_2,y_2)}~.
\end{equation}
The first inequality follows from Mrs. Gerber's Lemma. To see why the second inequality is true, note that
\begin{align}
\frac{\partial f}{\partial x} =& \frac{\partial f}{\partial p}  \frac{\partial p}{\partial x}\\
=&(1-2q)\log\left(\frac{1-p \star q}{p \star q}\right) \times \frac{1}{\log\left(\frac{1-p}{p}\right)}~,
\end{align}
where $x = h(p)$ and $y = h(q)$ with $0 \le p, q \le \frac{1}{2}$. Thus, for a fixed $p$, as $q$ increases $\frac{\partial f}{\partial x}$ strictly decreases, i.e. for fixed $x$, as $y$ increases $\frac{\partial f}{\partial x}$ strictly decreases. Note also that at least one of the two inequalities is strict as $(x_1,y_1) \neq (x_2,y_2)$. Thus
\begin{equation}\label{ineq1}
\frac{\partial f}{\partial x}\bigg|_{(x_1,y_1)} < \frac{\partial f}{\partial x}\bigg|_{(x_2,y_2)}~.
\end{equation}
It remains to consider the case $y_1 < y_2$. We can also assume $x_1 < x_2$, since $x_1 = x_2$ combined with $y_1 < y_2$ gives 
$$\frac{\partial f}{\partial x}\bigg|_{(x_1,y_1)} > \frac{\partial f}{\partial x}\bigg|_{(x_2,y_2)}~.$$
The only remaining case is thus $(x_1,y_1) < (x_2,y_2)$. Consider the line passing through the origin and $(x_1,y_1)$. We again break this up into two cases: either $y_2 \geq x_2 \frac{y_1}{x_1}$ or  $y_2  \leq x_2 \frac{y_1}{x_1}$.  \\
If  $y_2 \geq x_2 \frac{y_1}{x_1}$, 
\begin{equation}
\frac{\partial f}{\partial x}\bigg|_{(x_1,y_1)} > \frac{\partial f}{\partial x}\bigg|_{(x_2, x_2 \frac{y_1}{x_1})} \geq  \frac{\partial f}{\partial x}\bigg|_{(x_2,y_2)}~,
\end{equation}
where the first inequality follows from Lemma \ref{dfdx}, and the second follows from $\frac{\partial f}{\partial x}$ decreasing for a fixed $x$ and an increasing $y$.\\
If   $y_2  \leq x_2 \frac{y_1}{x_1}$, we have
\begin{equation}
\frac{\partial f}{\partial y}\bigg|_{(x_1,y_1)} > \frac{\partial f}{\partial y}\bigg|_{(y_2\frac{x_1}{y_1}, y_2)} \geq  \frac{\partial f}{\partial y}\bigg|_{(x_2,y_2)}~,
\end{equation}
where the first inequality follows from Lemma \ref{dfdx} and the fact that $y_2\frac{x_1}{y_1} > x_1$. The second inequality follows from the symmetric analogue of $\frac{\partial f}{\partial x}$ decreasing for a fixed $x$ and an increasing $y$, which is that $\frac{\partial f}{\partial y}$ decreases for a fixed $y$ and an increasing $x$.. This completes the proof of Lemma \ref{unique}.
\end{proof}

\section{An EPI and MGL for $\mathbb{Z}_4$-valued random variables}

Analogous to the framework for Shannon's EPI in the case of continuous random variables, we consider two independent random variables $X$ and $Y$ taking values in the cyclic group $\mathbb{Z}_4$ and seek to determine the minimum possible entropy of the random variable $X+Y$, where $+$ stands for the group addition, and we a priori fix the entropy of $X$ and that of $Y$.

Formally, we define $f_4: [0,\log4] \times [0,\log4] \to [0,\log4]$ by
\begin{equation}\label{f4}
f_4(x,y) = \min_{H(X) = x, H(Y)=y} H(X+Y)~.
\end{equation}
Thus $f_4 = f_{\mathbb{Z}_4}$, where $f_{\mathbb{Z}_4}$ is defined in equation (\ref{fgdefn}). In this section we 
will also use the notation $f_2$ for $f_{\mathbb{Z}_2}$,
so we have $f_2: [0,\log2] \times [0,\log2] \to [0,\log2]$
given by
\begin{equation}\label{f2}
f_2(x,y) = \min_{H(X) = x, H(Y)=y} H(X+Y)~.
\end{equation}

We will prove:
\begin{theorem} \label{z4}
$$f_4(x,y) = \begin{cases} x, & \mbox{if } \log2 \leq x \leq \log4, 0 \leq y \leq \log2~, \\ y, & \mbox{if } 0 \leq x \leq \log2, \log2 \leq y \leq \log4~,\\ f_2(x,y), & \mbox{if } 0 \leq x,y \leq \log2~,\\f_2(x-\log2,y-\log2) + \log2, & \mbox{ if } \log2 \leq x,y \leq \log4~. \end{cases}$$
\end{theorem}

The following corollary is immediate from Theorem \ref{z4} and Mrs. Gerber's Lemma.
\begin{corollary}\label{mgl_z4}
$f_4(x,y)$ is convex in $x$ for a fixed $y$, and by symmetry convex in $y$ for a fixed $x$.
\end{corollary}
\begin{proof}[Proof of Theorem \ref{z4}]
We deal with the initial two cases first. Without loss of generality, assume $ \log2 \leq x \leq \log4, 0 \leq y \leq \log2$. Note that we have the trivial lower bound
\begin{equation}
f_4(x,y) \geq x
\end{equation}
obtained from $H(X+Y) \geq H(X)$. 
Thus if we can find distributions for $X$ and $Y$ such that this lower bound is achieved, then it implies $f_4(x,y) = x$. This is exactly what we do. Since $y \leq \log2$, let $\beta = h^{-1}(y)$ and consider the distribution of $Y$
$$p_Y := (\beta, 0, 1-\beta, 0)~.$$
Also,  as $ \log2 \leq x$, we can find $\alpha$  such that $\log2 + H(2\alpha,1- 2\alpha) = x$. Using this $\alpha$, define
$$p_X := (\alpha, 1-\alpha, \alpha, 1-\alpha)~.$$
The distribution of $X+Y$ is given by the cyclic convolution $p_X \circledast_4 p_Y$, which in this case is $p_X$ again. Thus $H(X+Y) = H(X)$, and $f_4(x,y) = x$.\\
\smallskip

Before starting on the other two cases, we derive some preliminary inequalities. We'll think of distributions on $\mathbb{Z}_4$ as a combination of distributions supported on $\{0,2\}$ and $\{1,3\}$. For a random variable $X$, we write its distribution $p_X$ as
$$p_X = \alpha(p_0,0,p_2,0) + (1-\alpha)(0,p_1,0,p_3) = (\alpha p_0, (1-\alpha)p_1, \alpha p_2, (1-\alpha)p_3)~,$$
where $1 \geq p_0,p_1,p_2,p_3, \alpha \geq 0$ and
\begin{align*}
p_0 + p_2 &= 1~,\\
p_1+p_3 &= 1~.
 \end{align*}
 Similary we write 
 $$p_Y = \beta(q_0,0,q_2,0) + (1-\beta)(0,q_1,0,q_3) = (\beta q_0, (1-\beta)q_1, \beta q_2, (1-\beta)q_3)~,$$
 where $1 \geq q_0,q_1,q_2,q_3, \beta \geq 0$ and
\begin{align*}
q_0 + q_2 &= 1~,\\
q_1+q_3 &= 1~.
 \end{align*}
Let $X+Y = Z$. The distribution of $Z$ is given by
\begin{align}
p_Z &= p_X \circledast_4 p_Y\\
&= \bigg(\alpha(p_0,0,p_2,0) + (1-\alpha)(0,p_1,0,p_3)\bigg) \circledast_4 \bigg(\beta(q_0,0,q_2,0) + (1-\beta)(0,q_1,0,q_3)\bigg)\\
&= \bigg(\alpha\beta(p_0,0,p_2,0) \circledast_4 (q_0,0,q_2,0) +  (1-\alpha)(1-\beta)(0,p_1,0,p_3)\circledast_4(0,q_1,0,q_3)\bigg)\\
&+  \bigg(\alpha(1-\beta)(p_0,0,p_2,0) \circledast_4 (0,q_1,0,q_3) +  (1-\alpha)\beta(0,p_1,0,p_3)\circledast_4(q_0,0,q_2,0)\bigg)~. \nonumber
\end{align}
Thus
\begin{align}
H(p_Z) &= h(\alpha \star \beta) \label{h}\\
&+ (1-\alpha\star\beta)H\left(\frac{\alpha\beta}{1-\alpha\star\beta }(p_0,p_2) \circledast_2 (q_0,q_2) + \frac{(1-\alpha)(1-\beta)}{1-\alpha\star\beta} (p_1,p_3)\circledast_2(q_1,q_3)\right) \nonumber\\
&+ (\alpha\star\beta)H\left(\frac{\alpha(1-\beta)}{\alpha\star\beta }(p_0,p_2) \circledast_2 (q_1,q_3) + \frac{(1-\alpha)\beta}{\alpha\star\beta} (p_1,p_3)\circledast_2(q_0,q_2)\right) \nonumber\\
&\geq h(\alpha\star\beta) + \alpha\beta H\bigg((p_0,p_2) \circledast_2 (q_0,q_2)\bigg) + (1-\alpha)(1-\beta)H\bigg((p_1,p_3)\circledast_2(q_1,q_3)\bigg) \label{concaveh}\\ 
 &+ \alpha(1-\beta)H\bigg((p_0,p_2) \circledast_2 (q_1,q_3)\bigg) + (1-\alpha)\beta H\bigg((p_1,p_3)\circledast_2(q_0,q_2)\bigg) \nonumber \\
 &= f_2\bigg(h(\alpha), h(\beta)\bigg) + \alpha\beta f_2\bigg(H(p_0,p_2),H(q_0,q_2)\bigg) +  (1-\alpha)(1-\beta)f_2\bigg(H(p_1,p_3),H(q_1,q_3)\bigg) \label{deff2}\\ 
 &+ \alpha(1-\beta)f_2\bigg(H(p_0,p_2), H(q_1,q_3)\bigg) + (1-\alpha)\beta f_2\bigg(H(p_1,p_3),H(q_0,q_2)\bigg) \nonumber\\
 &\geq  f_2\bigg(h(\alpha), h(\beta)\bigg) + \alpha f_2\bigg(H(p_0,p_2), \beta H(q_0,q_2) + (1-\beta)H(q_1,q_3) \bigg) \label{convexf21}\\ 
 &+ (1-\alpha)f_2\bigg(H(p_1,p_3), \beta H(q_0,q_2) + (1-\beta)H(q_1,q_3) \bigg) \nonumber\\
 &\geq f_2\bigg(h(\alpha), h(\beta)\bigg) + f_2\bigg(\alpha H(p_0,p_2) + (1-\alpha)H(p_1,p_3),  \beta H(q_0,q_2) + (1-\beta)H(q_1,q_3)\bigg) \label{convexf22}\\
 &= f_2\bigg(h(\alpha), h(\beta)\bigg) + f_2\bigg(H(X)-h(\alpha), H(Y)-h(\beta)\bigg)~.
 \end{align}
 In this sequence of inequalities, (\ref{h}) is a simple expansion of entropy, (\ref{concaveh}) is got via concavity of entropy, (\ref{deff2}) is simply a restatement in terms of $f_2$, (\ref{convexf21}) and (\ref{convexf22}) are obtained using convexity in Mrs. Gerber's Lemma, and the last equality follows from the chain rule of 
entropy.
\smallskip

Coming back to the remaining two cases of Theorem \ref{z4}, we can write down the following inequalities as consequences of the above inequalities:\\
For $0 \leq x,y \leq \log2$
\begin{equation}\label{c3}
f_4(x,y) \geq \min_{u,v} f_2(u,v) + f_2(x-u,y-v)~,
\end{equation}
where $0 \leq u \leq x$ and $0 \leq v \leq y$.\\
 For $\log2 \leq x,y \leq \log4$,
\begin{equation}\label{c4}
f_4(x,y) \geq \min_{u,v} f_2(u,v) + f_2(x-u,y-v)~,
\end{equation}
where $x-\log2 \leq u \leq 1$ and $y - \log2 \leq v \leq 1$. \\

\smallskip

Consider the third case, $0 \leq x,y \leq \log2$. We'll show that the minimum in (\ref{c3}) is when $u, v$ are both equal to $0$ (or by symmetry $u=x$, $v=y$) and the value of the minimum is $f_2(x,y)$.\\ 

We'll first prove a small claim.

\begin{claim}\label{dfdx<1}
$\bigg|\frac{\partial f_2}{\partial x}\bigg | \leq 1$, with strict inequality if $(x,y)$ lies in the interior of the square $[0,\log2]\times[0,\log2]$.
\end{claim}  

\remark{By symmetry, $\bigg|\frac{\partial f_2}{\partial y}\bigg | \leq 1$, with strict inequality in the interior.}

\begin{proof}[Proof of Claim \ref{dfdx<1}] 
We note that when $y=0$, $f_2(x,0) = x$ which gives $\bigg|\frac{\partial f_2}{\partial x}\bigg |_{(x,0)} = 1$. Now fix $y > 0$. By Mrs. Gerber's Lemma, we know that $f_2(x,y)$ is convex is $x$ for a fixed $y$. This means that $\frac{\partial f_2}{\partial x}$ increases with $x$ and is maximum when $x=1$. Writing $x = h(p)$ and $y = h(q)$ with $0 \le p, q \le \frac{1}{2}$, we have $f_2(x,y) = h(p \star q)$, and
\begin{equation}
\frac{\partial f_2}{\partial x} =(1-2q)\log\left(\frac{1-p \star q}{p \star q}\right) \times \frac{1}{\log\left(\frac{1-p}{p}\right)}~.
\end{equation}
Taking the limit as $x \to 1$ is the same as taking the limit as $p \to \frac{1}{2}$. Using L'H\^{o}pital's rule, we get
\begin{equation}
\lim_{p \to \frac{1}{2}}(1-2q)\log\left(\frac{1-p \star q}{p \star q}\right) \times \frac{1}{\log\left(\frac{1-p}{p}\right)} = \lim_{p \to \frac{1}{2}} (1-2q)^2\frac{p(1-p)}{(p\star q)(1- p\star q)}~.
\end{equation} 
This is easily seen to be $(1-2q)^2$ which has magnitude $< 1$ for $ q \neq 0$. This establishes the claim.
\end{proof}

Now for  $0 \leq x,y \leq \log2$, consider the function $g: [0,x] \times [0,y] \to \mathbb{R}$ given by
$$g(u,v) := f_2(u,v) + f_2(x-u,y-v)~.$$
As per (\ref{c3}), we want to minimise $g$ over its domain. We can think of the domain as a rectangle with corner points $(0,0)$ and $(x,y)$ in $\mathbb{R}^2$. Suppose the minimum is achieved strictly in the interior of this rectangle, at a point say $(u^\star, v^\star)$, then we must have
\begin{align}
\frac{\partial g}{\partial u}\bigg|_{(u^\star, v^\star)} &= 0~,\\
\frac{\partial g}{\partial v}\bigg|_{(u^\star, v^\star)} &= 0~,
\end{align}
which implies
\begin{align}
\frac{\partial f_2}{\partial u}\bigg|_{(u^\star, v^\star)} = \frac{\partial f_2}{\partial u}\bigg|_{(x-u^\star,y- v^\star)}~,\\
\frac{\partial f_2}{\partial v}\bigg|_{(u^\star, v^\star)} = \frac{\partial f_2}{\partial v}\bigg|_{(x-u^\star,y- v^\star)}~.
\end{align}
By Lemma \ref{unique}, we infer that
\begin{equation}
(u^\star, v^\star) = (x-u^\star,y- v^\star)~.
\end{equation}
Thus $(u^\star, v^\star) = \left(\frac{x}{2},\frac{y}{2}\right)$. Now let $\theta = \frac{y}{x}$, and consider the function $g$ over the line with slope $\theta$ passing through the origin. By Lemma \ref{concavity}, we know that $f_2(t, \theta t)$ is concave, and thus so is $f_2(x -t, y-\theta t)$ and so is their addition $g(t,\theta t)$. Thus, the minimum value of $g(t, \theta t)$ must be attained at the extreme points and not in the interior. Note that since $(u^\star, v^\star)$ lies on this line, it cannot be the global minimum of $g$ on its domain. This leads us to conclude that the global minimum of $g$ is not attained anywhere in the interior of the rectangle and therefore must be attained on the boundary.\\
\smallskip

Now consider a point $(u_0,0)$ along the boundary. Taking the partial derivative with respect to $u$, 
\begin{equation}
\frac{\partial g}{\partial u}\bigg|_{(u_0,0)} = \frac{\partial f_2}{\partial u}\bigg|_{(u_0,0)} - \frac{\partial f_2}{\partial u}\bigg|_{(x-u_0,y)}  > 0~,\\
\end{equation}
where the inequality follows from $\frac{\partial f_2}{\partial u}\bigg|_{(u_0,0)}  = 1$ and $ \frac{\partial f_2}{\partial u}\bigg|_{(x-u_0,y)} < 1$ by Claim \ref{dfdx<1}. Similarly, for a boundary point of the form $(0,v_0)$ we can see that $\frac{\partial g}{\partial v}\bigg|_{(0,v_0)} > 0$. Hence, we conclude that the minimum value on the boundary is attained when $u=0,v=0$ and the value  is $f_2(x,y)$. Thus inequality (\ref{c3}) reduces to
\begin{equation}
f_4(x,y) \geq f_2(x,y)~.
\end{equation}
Clearly, $f_2(x,y)$ is achieved if the random variables are supported on the $\{0,2\}$, and therefore we get
\begin{equation}
f_4(x,y) = f_2(x,y) \mbox{ for } 0 \leq x,y \leq \log2~.
\end{equation}
This completes the proof for the third case.\\
\smallskip

Moving on to the last case, define $\tilde{u} = u - (x - \log2)$ and $\tilde{v} = v - (y-\log2)$. Rewriting (\ref{c4}), 
\begin{equation}
f_4(x,y) \geq \min_{\tilde{u},\tilde{v}} f_2(\log2 - \tilde{u}, \log2 - \tilde{v}) + f_2(\tilde{u}+(x - \log2),\tilde{v}+(y - \log2))~,
\end{equation}
where $$0 \leq \tilde{u} \leq 2\log2 - x~,$$  $$0 \leq \tilde{v} \leq 2\log2-y~.$$
Just as in the previous case, define 
$$g(\tilde{u},\tilde{v}) := f_2(\log2 - \tilde{u}, \log2 - \tilde{v}) + f_2(\tilde{u}+(x - \log2),\tilde{v}+(y - \log2))~.$$
The domain of $(\tilde{u},\tilde{v})$ can be thought of as a rectangle in $\mathbb{R}^2$ with corner points $(0,0), (2\log2 - x, 2\log2-y)$. Suppose the minimum value is attained at $(\tilde{u}^\star, \tilde{v}^\star)$ lying in the interior of this rectangular domain. In such a case we must have
\begin{align}
\frac{\partial g}{\partial \tilde{u}}\bigg|_{(\tilde{u}^\star, \tilde{v}^\star)} &= 0~,\\
\frac{\partial g}{\partial \tilde{v}}\bigg|_{(\tilde{u}^\star, \tilde{v}^\star)} &= 0~,
\end{align} 
which implies
\begin{align}
\frac{\partial f_2}{\partial \tilde{u}}\bigg|_{(\log2 - \tilde{u^\star}, \log2 - \tilde{v}^\star)} = \frac{\partial f_2}{\partial \tilde{u}}\bigg|_{(\tilde{u}^\star+(x - \log2),\tilde{v}^\star+(y - \log2))}~,\\
\frac{\partial f_2}{\partial \tilde{v}}\bigg|_{(\log2 - \tilde{u^\star}, \log2 - \tilde{v}^\star)} = \frac{\partial f_2}{\partial \tilde{v}}\bigg|_{(\tilde{u}^\star+(x - \log2),\tilde{v}^\star+(y - \log2))}~.
\end{align}
By Lemma \ref{unique}, we infer that
\begin{equation}
(\log2 - \tilde{u^\star}, \log2 - \tilde{v}^\star) = (\tilde{u}^\star+(x - \log2),\tilde{v}^\star+(y - \log2))~,
\end{equation}
which implies $(\tilde{u}^\star, \tilde{v}^\star) = \left(\log2 - \frac{x}{2}, \log2 - \frac{y}{2}\right)$. Now if $(\tilde{u}^\star, \tilde{v}^\star)$ were indeed the global minimum, then we must have for every $\theta$
\begin{equation}
\frac{d^2 }{dt^2}g(\tilde{u}^\star + t, \tilde{v}^\star+\theta t)\bigg|_{t=0} \geq 0~.
\end{equation} 
Note that
\begin{equation}
g(\tilde{u}^\star + t, \tilde{v}^\star+\theta t) = f_2\left(\frac{x}{2}-t, \frac{y}{2}-\theta t\right) +  f_2\left(\frac{x}{2}+t, \frac{y}{2}+\theta t\right)~. 
\end{equation}
Now choose $\theta = \frac{y}{x}$. By Lemma \ref{concavity}, we have
\begin{equation}
\frac{d^2 }{dt^2}g(\tilde{u}^\star + t, \tilde{v}^\star+\theta t)\bigg|_{t=0} = 2\frac{d^2 }{dt^2}f_2\left(\frac{x}{2}-t, \frac{y}{2}-\theta t\right)\bigg|_{t=0} < 0~.
\end{equation}
This means that  $(\tilde{u}^\star, \tilde{v}^\star)$ cannot be the global minimum, and the global minimum therefore must lie on the boundary. For a boundary point of the form $(\tilde{u}_0, 0)$ 
\begin{equation}
\frac{\partial g}{\partial \tilde{u}}\bigg|_{(\tilde{u}_0,0)} = -\frac{\partial f_2}{\partial \tilde{u}}\bigg|_{(\log2 - \tilde{u}_0,\log{2})} + \frac{\partial f_2}{\partial \tilde{u}}\bigg|_{(\tilde{u}_0 + x - \log2, y - \log2)}  > 0~,\\
\end{equation}
where the inequality follows from $\frac{\partial f_2}{\partial \tilde{u}}\bigg|_{(\log2 - \tilde{u}_0,\log{2})} = 0$ and $\frac{\partial f_2}{\partial \tilde{u}}\bigg|_{(\tilde{u}_0 + x - \log2, y - \log2)} > 0$. Similarly for a boundary point of the form $(0, \tilde{v}_0)$ we have
\begin{equation}
\frac{\partial g}{\partial \tilde{v}}\bigg|_{(0, \tilde{v}_0)} = -\frac{\partial f_2}{\partial \tilde{v}}\bigg|_{(\log2, \log2 - \tilde{v}_0)} + \frac{\partial f_2}{\partial \tilde{v}}\bigg|_{(x - \log2, \tilde{v}_0 + y - \log2)}  > 0~.\\
\end{equation}
In both cases we see that the minimum is attained when $(\tilde{u}, \tilde{v}) = (0,0)$ and the value of the minimum is $$f_2(\log2, \log2) + f_2(x - \log2, y - \log2) = \log2 + f_2(x-\log2, y-\log2)~.$$
Thus inequality (\ref{c4}) reduces to
\begin{equation}\label{c4bound}
f_4(x,y) \geq  \log2 + f_2(x-\log2, y-\log2)~.
\end{equation}
Since $\log2 \leq x,y \leq \log4$, we can find distributions $p_X = (\frac{\alpha}{2}, \frac{1-\alpha}{2}, \frac{\alpha}{2}, \frac{1-\alpha}{2})$ and $p_Y = (\frac{\beta}{2}, \frac{1-\beta}{2}, \frac{\beta}{2}, \frac{1-\beta}{2})$ such that $H(p_X) = x$ and $H(p_Y) = y$, with 
\begin{align}
x &= h(\alpha) + \log2~,\\
y &= h(\beta) + \log2~,\\
H(X+Y) = H(p_X \circledast_4 p_Y) &= H\left(\frac{1-\alpha \star \beta}{2}, \frac{\alpha \star \beta}{2}, \frac{1-\alpha \star \beta}{2}, \frac{\alpha \star \beta}{2}  \right)\\
&= h(\alpha \star \beta) + \log2 \nonumber \\
&= f_2(h(\alpha), h(\beta)) + \log2 \nonumber\\
&=  f_2(x-\log2, y-\log2) + \log2~.
\end{align}
Thus the bound in (\ref{c4bound}) is achieved, and we conclude that
\begin{equation}
f_4(x,y) =  \log2 + f_2(x-\log2, y-\log2)~.
\end{equation}
This completes the proof of Theorem \ref{z4}.
\end{proof}

\begin{proof}[Proof of Corollary \ref{mgl_z4}]
Consider the function $f_x(y) = f(x,y)$. We look at two cases, $0 \leq x \leq \log2$ and $\log2 \leq x \leq \log4$. In the first case, 
$$f_x(y) = \begin{cases} f_2(x,y), & \mbox{if } 0 \leq y \leq \log2~,\\ y & \mbox{if } \log2 \leq y \leq \log4~. \end{cases}$$
Now $f_2(x,y)$ for a fixed $x$ and $0 \leq y \leq \log2$ is convex by MGL, and for values of $y$ beyond $\log2$ the function $f_x$ is linear with slope $1$. By Claim \ref{dfdx<1}, attaching this linear part to a convex function will not affect the convexity since the slope of the linear part ($=1$) is greater than or equal to the derivative of the convex part. Similarly for the second case,
$$f_x(y) = \begin{cases} x, & \mbox{if } 0 \leq y \leq \log2~,\\ f_2(x-\log2,y-\log2)+\log2 & \mbox{if } \log2 \leq y \leq \log4~. \end{cases}$$
This too, has a linear part with slope $0$ attached before a convex part with slope greater equal $0$ everywhere, thus the overall function continues being convex.
\end{proof}

\section{An EPI and MGL for $\mathbb{Z}_{2^n}$ valued random variables}
Analogous to the $\mathbb{Z}_4$ case, we consider two independent random variables $X$ and $Y$ taking values in the cyclic group $\mathbb{Z}_{2^n}$ and seek to determine the minimum possible entropy of the random variable $X+Y$ where $+$ stands for the group addition, where we a priori fix the value of the $H(X)$ and
the value of $H(Y)$.

Formally, we define $f_{2^n}: [0,n\log2] \times [0,n\log2] \to [0,n\log2]$ by
\begin{equation}\label{f2n}
f_{2^n}(x,y) = \min_{H(X) = x, H(Y)=y} H(X+Y)~.
\end{equation}
Thus $f_{2^n} = f_{\mathbb{Z}_{2^n}}$, where 
$f_{\mathbb{Z}_{2^n}}$ is defined in equation (\ref{fgdefn}).

$f_{2^n}$ is completely determined in the following theorem:
\begin{theorem}\label{z2n}
$$f_{2^n}(x,y) = \begin{cases} f_{2}(x-k\log2, y-k\log2)+k\log2, & \mbox{if } k\log2 \leq x,y \leq (k+1)\log2~,\\  \max(x,y) & \mbox{otherwise. } \end{cases}$$
\end{theorem}

The following corollary is an immediate consequence of Theorem
\ref{z2n} and Mrs. Gerber's Lemma:
\begin{corollary}\label{mgl_z2n}
$f_{2^n}(x,y)$ is convex in $x$ for a fixed $y$, and by symmetry is convex in $y$ for a fixed $x$.
\end{corollary}

\begin{proof}[Proof of Theorem \ref{z2n}]
We deal with the second case first. Assume $$ k_1\log2 \leq x \leq (k_1+1)\log2~,$$ $$k_2\log2  \leq y \leq (k_2+1)\log2~,$$ where $k_1 \neq k_2$. Without loss of generality, assume $k_1 > k_2$. Note that we have the trivial lower bound
\begin{equation}
f_{2^n}(x,y) \geq x
\end{equation}
obtained from $H(X+Y) \geq H(X)$. 
Thus if we can find distributions for $X$ and $Y$ such that this lower bound is achieved, then this would imply that $f_{2^n}(x,y) = x$. This is exactly what we do. Since $y \leq k_1\log2$, let the distribution of $Y$ be any distribution supported on the subgroup $\mathbb{Z}_{2^{k_1}}$ which is contained in $\mathbb{Z}_{2^n}$ such that $H(p_Y) = y$. Here as usual the subgroup $\mathbb{Z}_{2^k}$ in $\mathbb{Z}_{2^n}$ is the set $\{0,2^{n-k}, 2. 2^{n-k},3.2^{n-k},..., (2^k - 1)2^{n-k}\}$. Also,  as $k_1 \log2 \leq x \leq (k_1+1)\log2$, we can find a distribution of $X$ which is supported on the subgroup $\mathbb{Z}_{2^{k_1+1}}$ and is constant over the cosets  $\mathbb{Z}_{2^{k_1+1}}/ \mathbb{Z}_{2^{k_1}}$.  The distribution of $X+Y$ is given by the cyclic convolution $p_X \circledast_{2^n} p_Y$ which in this case is $p_X$ again. Thus $H(X+Y) = H(X)$, and $f_{2^n}(x,y) = x$. \\
\smallskip

Before considering the remaining case, we derive some preliminary inequalities. We'll use induction, assume that the theorem and the corollary is true for $2^{n-1}$ and prove it for $2^n$. We'll think of distributions on $\mathbb{Z}_{2^n}$ as a combination of distributions supported on the cosets of $\mathbb{Z}_{2^{n-1}}$ in $\mathbb{Z}_{2^n}$. For a random variable $X$, we can write 
$$p_X = \alpha p_E + (1-\alpha)p_O~,$$
where $1 \geq  \alpha \geq 0$, with $p_E$ supported only on the subgroup $\mathbb{Z}_{2^{n-1}}$ of $\mathbb{Z}_{2^n}$ and $p_O$ supported on the remaining half of $\mathbb{Z}_{2^n}$.
Similary we write 
 $$p_Y = \beta q_E + (1-\beta) q_O~, $$
 where $1 \geq \beta \geq 0$.

Let $X+Y = Z$. The distribution of $Z$ is given by
\begin{align}
p_Z &= p_X \circledast_{2^n} p_Y\\
&= \bigg(\alpha p_E+ (1-\alpha) p_O \bigg) \circledast_{2^n} \bigg(\beta q_E + (1-\beta)q_O\bigg)\\
&= \bigg(\alpha\beta p_E \circledast_{2^n} q_E +  (1-\alpha)(1-\beta)p_O \circledast_{2^n} q_O\bigg)\\
&+  \bigg(\alpha(1-\beta)p_E \circledast_{2^n} q_O +  (1-\alpha)\beta p_O \circledast_{2^n} q_E \bigg)~. \nonumber
\end{align}
Thus
\begin{align}
H(p_Z) &= h(\alpha \star \beta) \label{hn}\\
&+ (1-\alpha\star\beta)H\left(\frac{\alpha\beta}{1-\alpha\star\beta } p_E \circledast_{2^{n-1}} q_E + \frac{(1-\alpha)(1-\beta)}{1-\alpha\star\beta} p_O \circledast_{2^{n-1}}q_O \right) \nonumber\\
&+ (\alpha\star\beta)H\left(\frac{\alpha(1-\beta)}{\alpha\star\beta }p_E \circledast_{2^{n-1}} q_O + \frac{(1-\alpha)\beta}{\alpha\star\beta} p_O \circledast_{2^{n-1}} q_E \right) \nonumber\\
&\geq h(\alpha\star\beta) + \alpha\beta H\bigg(p_E \circledast_{2^{n-1}} q_E \bigg) + (1-\alpha)(1-\beta)H\bigg(p_O \circledast_{2^{n-1}} q_O\bigg) \label{concavehn}\\ 
 &+ \alpha(1-\beta)H\bigg(p_E \circledast_{2^{n-1}} q_O \bigg) + (1-\alpha)\beta H\bigg(p_O \circledast_{2^{n-1}} q_E \bigg) \nonumber \\
 &\geq f_2\bigg(h(\alpha), h(\beta)\bigg) + \alpha\beta f_{2^{n-1}}\bigg(H(p_E),H(q_E)\bigg) +  (1-\alpha)(1-\beta)f_{2^{n-1}}\bigg(H(p_O),H(q_O)\bigg) \label{deff2n}\\ 
 &+ \alpha(1-\beta)f_{2^{n-1}}\bigg(H(p_E), H(q_O)\bigg) + (1-\alpha)\beta f_{2^{n-1}}\bigg(H(p_O),H(q_E)\bigg) \nonumber\\
 &\geq  f_2\bigg(h(\alpha), h(\beta)\bigg) + \alpha f_{2^{n-1}}\bigg(H(p_E), \beta H(q_E) + (1-\beta)H(q_O) \bigg) \label{convexf21n}\\ 
 &+ (1-\alpha)f_{2^{n-1}}\bigg(H(p_O), \beta H(q_E) + (1-\beta)H(q_O) \bigg) \nonumber\\
 &\geq f_2\bigg(h(\alpha), h(\beta)\bigg) + f_{2^{n-1}}\bigg(\alpha H(p_E) + (1-\alpha)H(p_O),  \beta H(q_E) + (1-\beta)H(q_O)\bigg) \label{convexf22n}\\
 &= f_2\bigg(h(\alpha), h(\beta)\bigg) + f_{2^{n-1}}\bigg(H(X)-h(\alpha), H(Y)-h(\beta)\bigg)~.
 \end{align}
 In this sequence of inequalities, (\ref{hn}) is a simple expansion of entropy, (\ref{concavehn}) is got via concavity of entropy, (\ref{deff2n}) is using the definition of $f$, (\ref{convexf21n}) and (\ref{convexf22n}) are obtained using Mrs. Gerber's Lemma for $2^{n-1}$ (by induction hypothesis), and the last equality follows from the chain rule of entropy.
\smallskip

Coming back to the remaining cases of Theorem \ref{z2n}, we can write down the following inequalities as consequences of the preceding sequence of inequalities:\\
For $0 \leq x,y \leq \log2$, 
\begin{equation}\label{c3n}
f_{2^n}(x,y) \geq \min_{u,v} f_2(u,v) + f_{2^{n-1}}(x-u,y-v)~,
\end{equation}
where $0 \leq u \leq x$ and $0 \leq v \leq y$.\\
 For $(n-1)\log2 \leq x,y \leq n\log2$,
\begin{equation}\label{c4n}
f_{2^n}(x,y) \geq \min_{u,v} f_2(u,v) + f_{2^{n-1}}(x-u,y-v)~,
\end{equation}
where $x-(n-1)\log2 \leq u \leq \log2$ and $y - (n-1)\log2 \leq v \leq \log2$. \\
For $(k-1)\log2 \leq x,y \leq k\log2$, $k \neq 1, n$,
\begin{equation}\label{c5n}
f_{2^n}(x,y) \geq \min_{u,v} f_2(u,v) + f_{2^{n-1}}(x-u,y-v)~,
\end{equation}
where $0 \leq u \leq \log2$ and $0  \leq v \leq \log2$. \\

\smallskip

We'll consider the above three cases separately and prove the theorem in each of those three cases.
\begin{claim}\label{case1}
For $0 \leq x,y \leq \log2$ we have
$$f_{2^n}(x,y) = f_2(x,y)~.$$
\end{claim}
\begin{proof}[Proof of Claim \ref{case1}]
From equation (\ref{c3n}) we have 
$$f_{2^n}(x,y) \geq \min_{u,v} f_2(u,v) + f_{2^{n-1}}(x-u,y-v)~,$$
where the maximum is over $0 \leq u,v \leq \log 2$.
However, by our induction hypothesis 
$$f_2(u,v) + f_{2^{n-1}}(x-u,y-v) = f_2(u,v) + f_2(x-u,y-v)~,$$
and from the proof of the $\mathbb{Z}_4$ case, the value of this minimum is $f_2(x,y)$. Since this value is clearly achieved, we have $f_{2^n}(x,y) = f_2(x,y)$.
\end{proof}

\begin{claim}\label{case2}
For $(k-1)\log2 \leq x,y \leq k\log2$, $k \neq 1, n$, we have
$$f_{2^n}(x,y) = (k-1)\log2 + f_2(x - (k-1)\log2, y-(k-1)\log2)~.$$
\end{claim}
\begin{proof}[Proof of Claim \ref{case2}]
From (\ref{c4n}) we have 
$$f_{2^n}(x,y) \geq \min_{u,v} f_2(u,v) + f_{2^{n-1}}(x-u,y-v)~.$$
We first note that if the minimum of the above expression occurs at $(u^\star, v^\star)$ then we must have
\begin{equation}\label{case2.1}
(k-1)\log2 \leq x-u^\star,y-v^\star \leq k\log2~,
\end{equation}
or
\begin{equation}\label{case2.2}
(k-2)\log2 \leq x-u^\star,y-v^\star \leq (k-1)\log2~.
\end{equation}
To see this, suppose that w.l.o.g. we have $$(k-2)\log2 <  x-u^\star < (k-1)\log2~,$$  $$(k-1)\log2 < y-v^\star < k\log2~.$$
Let $\tilde{u}$ be such that $x - \tilde{u} = (k-1)\log2$. We have $\tilde{u} < u^\star$. By induction hypothesis, 
$$f_{2^{n-1}}(x - u^\star, y - v^\star) = f_{2^{n-1}}(x - \tilde{u}, y - v^\star) = y - v^\star~.$$ But since $\tilde{u} < u^\star$ we also have
$$f_2(\tilde{u}, v^\star) <  f_2(u^\star, v^\star)~.$$ 
This leads us to conclude that
$$f_2(\tilde{u}, v^\star) +  f_{2^{n-1}}(x - \tilde{u}, y - v^\star)  < f_2(u^\star, v^\star) + f_{2^{n-1}}(x - u^\star, y - v^\star)~,$$
which contradicts $(u^\star, v^\star)$ being the minimizer.
Now suppose we minimize over all pairs $u,v$ such that (\ref{case2.1}) holds. By induction hypothesis, 
\begin{align*}
\min_{u,v} f_2(u,v) + f_{2^{n-1}}(x-u,y-v) &= \min_{u,v} f_2(u,v) + f_2\bigg(x-u - (k-1)\log2,y-v - (k-1)\log2\bigg)\\
 & ~~~~~+ (k-1)\log2\\
&= (k-1)\log2 +  \min_{u,v} f_2(u,v) + f_2\bigg(x- (k-1)\log2 - u,y- (k-1)\log2 - v\bigg)~.
\end{align*}
From the proof of the $\mathbb{Z}_4$ case, we have that the minimum of the above expression is when $u, v = 0$ which gives us 
\begin{equation}\label{case2.1.0}
\min_{u,v} f_2(u,v) + f_{2^{n-1}}(x-u,y-v) = (k-1)\log2 + f_2\bigg(x-(k-1)\log2, y-(k-1)\log2\bigg)~,
\end{equation}
where it is implicit that the minimization is taken over all pairs $u,v$ such that (\ref{case2.1}) holds.
\smallskip

Now we minimize over all pairs $u,v$ such that (\ref{case2.2}) holds. By induction hypothesis, 
\begin{align*}
\min_{u,v} f_2(u,v) + f_{2^{n-1}}(x-u,y-v) &= \min_{u,v} f_2(u,v) + f_2\bigg(x-u - (k-2)\log2,y-v - (k-2)\log2\bigg)\\
 & ~~~~~+ (k-2)\log2\\
&= (k-2)\log2 +  \min_{u,v} f_2(u,v) + f_2\bigg(x- (k-2)\log2 - u,y- (k-2)\log2 - v\bigg)~.
\end{align*}
Again, by the proof of the $\mathbb{Z}_4$ case we have that the minimum value of the above expression is attained when $u, v = \log2$. Substituting we get
\begin{align}
\min_{u,v} f_2(u,v) + f_{2^{n-1}}(x-u,y-v)  &= (k-2)\log2 + f_2(\log2, \log2) + f_2\bigg(x - (k-1)\log2, y-(k-1)\log2\bigg) \nonumber \\ 
&= (k-1)\log2 +  f_2\bigg(x - (k-1)\log2, y-(k-1)\log2\bigg)~. \label{case2.2.0}
\end{align}
Comparing (\ref{case2.1.0}) and (\ref{case2.2.0}) we arrive at 
\begin{align*}
f_{2^n}(x,y) &\geq \min_{u,v} f_2(u,v) + f_{2^{n-1}}(x-u,y-v)\\
 &= (k-1)\log2 +  f_2\bigg(x - (k-1)\log2, y-(k-1)\log2\bigg) \\&
 = f_{2^{n-1}}(x,y)~.
 \end{align*}
Since $f_{2^{n-1}}(x,y)$ is achieved by supporting $X$ and $Y$ on $\mathbb{Z}_{2^{n-1}}$ we have $f_{2^n}(x,y) = (k-1)\log2 +  f_2\bigg(x - (k-1)\log2, y-(k-1)\log2\bigg)$, proving the claim.
\end{proof}

\begin{claim}\label{case3}
For $(n-1)\log2 \leq x,y \leq n\log2$,
$$f_{2^n}(x,y) = (n-1)\log2 + f_2\bigg(x-(n-1)\log2, y-(n-1)\log2\bigg)~.$$
\end{claim}
\begin{proof}[Proof of Claim \ref{case3}]
We have 
$$f_{2^n}(x,y) \geq \min_{u,v} f_2(u,v) + f_{2^{n-1}}(x-u,y-v)~,$$
where $x-(n-1)\log2 \leq u \leq \log2$ and $y - (n-1)\log2 \leq v \leq \log2$. Using our induction hypothesis,  
\begin{align*}
\min_{u,v} f_2(u,v) + f_{2^{n-1}}(x-u,y-v) &=  \min_{u,v} f_2(u,v) + f_2\bigg(x-u-(n-2)\log2,y-v-(n-2)\log2\bigg)\\
&~~~~~~ + (n-2)\log2\\
&= (n-2)\log2 + f_2(\log2, \log2) + f_2\bigg(x - (n-1)\log2, y-(n-1)\log2\bigg)\\
&= (n-1)\log2 + f_2\bigg(x - (n-1)\log2, y-(n-1)\log2\bigg)~,
\end{align*}
where the second equality follows from the proof on $\mathbb{Z}_4$, where we had that the minimum of such an expression is attained when $u,v = \log2$. To show that equality is attained, consider $p_X$ such that it takes a constant value $\frac{\alpha}{2^{n-1}}$ on the subgroup of size $2^{n-1}$ of $\mathbb{Z}_{2^n}$ and a constant value $\frac{1-\alpha}{2^{n-1}}$ on the remaining half of $\mathbb{Z}_{2^n}$ such that $H(p_X) = x$. Similarly choose $\beta$ such that $p_Y$ takes a constant value $\frac{\beta}{2^{n-1}}$ on the subgroup of size $2^{n-1}$ of $\mathbb{Z}_{2^n}$ and a constant value $\frac{1-\beta}{2^{n-1}}$ on the remaining half of $\mathbb{Z}_{2^n}$, such that $H(p_Y) = y$. We have
$$x = H(p_X) = (n-1)\log2 + h(\alpha),$$
$$y = H(p_Y) = (n-1)\log2 + h(\beta).$$
It is easy to verify that $$H(X+Y) = h(\alpha \star \beta) + (n-1)\log2 = f_2(x-(n-1)\log2,y-(n-1)\log2)+(n-1)\log2.$$
This completes the proof of the claim.
\end{proof}

The above claims complete the proof of Theorem \ref{z2n}.
\end{proof}

\begin{proof}[Proof of Corollary \ref{mgl_z2n}]
Consider $k\log2 \leq x \leq (k+1)\log2$ and the function $f_x(y) = f_{2^n}(x,y)$. We have
$$f_x(y) = \begin{cases} x, & \mbox{if } 0 \leq y \leq k\log2~,\\ f_2(x-k\log2,y-k\log2)+k\log2 & \mbox{if } k\log2 \leq y \leq (k+1)\log2~,\\y & \mbox{if } (k+1)\log2 \leq y~. \end{cases}$$
This is immediately seen to be convex using MGL and Claim \ref{dfdx<1}.
\end{proof}

\section{An EPI and MGL for abelian groups of order $2^n$}
We first prove a lemma.

\begin{lemma}\label{lowerbound}
Consider two abelian groups $G$ and $H$ with the corresponding $f_G$ and $f_H$ functions, such that $f_G$ satisfies the generalized MGL. Then the following lower bound holds for $f_{G\oplus H}$:
\begin{equation}
f_{G\oplus H}(x,y) \geq \min_{u, v} f_H(u,v) + f_G(x-u,y-v)~,
\end{equation}
where $u,v$ vary over 
\begin{align}
&\max(0, x - \log |G|) \leq u \leq \min(\log |H|,x)~,\\
&\max(0, y - \log |G|) \leq v \leq \min(\log |H|,y)~.
\end{align}
\end{lemma}

\begin{proof}[Proof of Lemma \ref{lowerbound}]
We can write any probability distribution on $G \oplus H$ in terms of a convex combination of probability distributions supported on the cosets of $G$. Note that there will be $|H|$ such cosets. Suppose $X$ and $Y$ are random variables taking values in $G \oplus H$. We can write $p_X$ and $p_Y$ as
\begin{align}
p_X &= \sum_{h \in H} \alpha_h p_h~, \label{ph}\\
p _Y &= \sum_{h \in H} \beta_h q_h~, \label{qh}
\end{align}
where each $p_h$ is a distribution supported on the coset $(G,0)+(0,h)$. The distribution of $Z = X+Y$ can be broken down in a similar fashion as in (\ref{ph}), (\ref{qh}).
\begin{equation}
p_Z = \sum_{h \in H} \gamma_h r_h~. \label{rh}
\end{equation}
Here we have
\begin{align}
\gamma_h &= \sum_{i \in H} \alpha_i\beta_{h-i} = (\alpha \circledast_H \beta)_h~,\\
r_h &= \frac{\sum_{i \in H} (\alpha_i\beta_{h-i})(p_i \circledast_G q_{h-i})}{(\alpha \circledast_H \beta)_h}~.
\end{align}
Thus, using chain rule of entropy, we can write $H(Z)$ as
\begin{align}
H(Z) &= H(\gamma) + \sum_{h \in H}\gamma_h H(r_h)\\
&= H(\alpha \circledast_H \beta) + \sum_{h \in H}(\alpha \circledast_H \beta)_hH\left( \frac{\sum_{i \in H} (\alpha_i\beta_{h-i})(p_i \circledast_G q_{h-i})}{(\alpha \circledast_H \beta)_h}\right)\\
&\geq H(\alpha \circledast_H \beta) + \sum_h \sum_i (\alpha_i\beta_{h-i})H(p_i \circledast_G q_{h-i}) \label{lb1}\\
&\geq f_H(H(\alpha),H(\beta)) + \sum_h \sum_i (\alpha_i\beta_{h-i})f_G(H(p_i),H(q_{h-i})) \label{lb2}\\
&= f_H(H(\alpha),H(\beta)) + \sum_i \sum_h (\alpha_i\beta_{h-i})f_G(H(p_i),H(q_{h-i}))\\
&= f_H(H(\alpha),H(\beta)) +  \sum_i \alpha_i \left(\sum_h \beta_{h-i}f_G(H(p_i),H(q_{h-i}))\right)\\
&\geq f_H(H(\alpha),H(\beta)) +  \sum_i \alpha_i f_G \left(H(p_i), \sum_h \beta_{h-i}H(q_{h-i})\right)\label{lb3}\\
&=  f_H(H(\alpha),H(\beta)) +  \sum_i \alpha_i f_G \left(H(p_i), \sum_h \beta_{h}H(q_{h})\right)\\
&\geq  f_H(H(\alpha),H(\beta)) + f_G\left(\sum_h \alpha_{h}H(p_{h}), \sum_h \beta_{h}H(q_{h})\right)\label{lb4}\\
&= f_H(H(\alpha),H(\beta)) + f_G\bigg(H(X)-H(\alpha),H(Y)-H(\beta)\bigg)~.
\end{align}
Here (\ref{lb1}) follows from concavity of entropy, (\ref{lb2}) follows from the definition of $f$, (\ref{lb3}) and (\ref{lb4}) follow from $f_G$ satisfying the generalized MGL.
Using the above, we can get the lower bound
\begin{equation}
f_{G\oplus H}(x,y) \geq \min_{u, v} f_H(u,v) + f_G(x-u,y-v)~,
\end{equation}
where $u,v$ vary over 
\begin{align}
&\max(0, x - \log |G|) \leq u \leq \min(\log |H|,x)~,\\
&\max(0, y - \log |G|) \leq v \leq \min(\log |H|,y)~.
\end{align}
\end{proof}

\begin{theorem}\label{abelian}
If $G$ is an abelian group of order $2^n$, then $f_G(x,y) = f_{2^n}(x,y)$.
\end{theorem}
\begin{proof}[Proof of Theorem \ref{abelian}]
Assume $$ k_1\log2 \leq x \leq (k_1+1)\log2,$$ $$k_2\log2  \leq y \leq (k_2+1)\log2,$$ where $k_1 \neq k_2$. Without loss of generality, assume $k_1 > k_2$. Note that we have the trivial lower bound
\begin{equation}
f_{G}(x,y) \geq x
\end{equation}
obtained from $H(X+Y) \geq H(X)$. 
Thus if we can find distributions for $X$ and $Y$ such that this lower bound is achieved, then it implies $f_{G}(x,y) = x$. This is exactly what we do. Let $G_1$ be a subgroup of $G$ of size $2^{k_1+1}$, and let $G_2$ be a subgroup of $G_1$ of size $2^{k_1}$. Consider the cosets of $G_2$ with respect to $G_1$, call them $C_0 ( = G_2)$ and $C_1$. Now consider the distribution of $X$ as taking a constant value on $C_0$ and on $C_1$ such that $H(X) = x$. Let the distribution of $Y$ be any arbitrary distribution on $C_0$ such that $H(Y) = y$. Notice that (in terms of coset addition)
\begin{align*}
&C_0 + C_0 = C_0~,\\
&C_0 + C_1 = C_1 + C_0 = C_1~,\\
&C_1+C_1 = C_0.\\
\end{align*}
Since $Y$ is supported only on $C_0$, and $X$ is uniform on $C_0$ and $C_1$ it is easy to see that $X+Y$ is also uniform on $C_0$ and $C_1$ and in fact has the same distribution as that of $X$. This takes care of all cases when $k_1 \neq k_2$ and we can only concern ourselves with the case  $k_1 = k_2 =: k$.
\smallskip

Now either $G$ is a a cyclic group of size $2^n$, or $G$ can be written as a direct sum $H_1 \oplus H_2$ where $H_1$ and $H_2$ are themselves abelian of size $2^{l_i}, i = 1,2$ respectively. In the first case, there is nothing to prove. So assume the second case holds, and without loss of generality let $l_1 \leq l_2$. Our proof will proceed in two steps, in the first step we show that $f_G(x,y) \leq f_{2^n}(x,y)$ and in the second we show that $f_G(x,y) \geq f_{2^n}(x,y)$. We'll use induction in the second step, where we assume the theorem holds true for the smaller groups $H_1$ and $H_2$ and prove it for $G$.

\begin{claim}\label{abelianstep1}
$f_G(x,y) \leq f_{2^n}(x,y)$
\end{claim} 
\begin{proof}[Proof of Claim \ref{abelianstep1}]
As before, let $$ k\log2 \leq x \leq (k+1)\log2,$$ $$k\log2  \leq y \leq (k+1)\log2.$$ 
Consider a subgroup $G_1$ of size $2^{k+1}$, and a subgroup $G_2$ of $G_1$ of size $2^k$. Let $C_0$ and $C_1$ be the cosets of $G_2$ in $G_1$. We consider a distribution of $X$ which takes a constant value on  $\frac{x_0}{2^k}$ on $C_0$ and a constant value $\frac{1-x_0}{2^k}$ on  $C_1$ and has $H(X) =x$. Similarly consider a distribution of $Y$ which takes constant values $\frac{y_0}{2^k}$ on $C_0$ and $\frac{y_1}{2^k}$ on $C_1$ and has $H(Y) = y$.
We have
$$H(X) = x = h(x_0) + k \log2,$$
$$H(Y) = y = h(y_0) + k \log 2.$$
It is easy to verify that 
$$H(X+Y) = h(x_0 \star y_0) + k\log2 = f_2(x-k\log2,y-k\log2)+k\log2 = f_{2^n}(x,y).$$
By the definition of $f_G$, we  get $$f_G(x,y) \leq f_{2^n}(x,y).$$
\end{proof}
\begin{claim}\label{abelianstep2}
$f_G(x,y) \geq f_{2^n}(x,y)$
\end{claim}
\begin{proof}
By our assumptions, $G = H_1 \oplus H_2$ where $|H_1| = 2^{l_1}$, $|H_2| = 2^{l_2}$ where $l_1 + l_2 = n$ and without loss of generality $0 < l_1 \leq l_2$. We also assume that the theorem holds for $H_1$ and $H_2$ and prove it by induction for $G$. By Lemma \ref{lowerbound} we have the lower bound
\begin{equation}\label{gh1h2}
f_{H_1\oplus H_2}(x,y) \geq \min_{u, v} f_{H_2}(u,v) + f_{H_1}(x-u,y-v),
\end{equation}
where $u,v$ vary over 
\begin{align}
&\max(0, x - \log |H_1|) \leq u \leq \min(\log |H_2|,x), \label{condu}\\
&\max(0, y - \log |H_1|) \leq v \leq \min(\log |H_2|,y).\label{condv}
\end{align}
Note that (\ref{condu}) and (\ref{condv}) are equivalent, respectively to 
\begin{align}
\max(0,x-\log|H_2|) \leq (x-u) \leq \min(\log|H_1|,x) \tag{24a}, \label{condua}\\
\max(0,y-\log|H_2|) \leq (y-v) \leq \min(\log|H_1|,y) \tag{25a} \label{condva}.
\end{align}
To facilitate the discussion, we term as a `diagonal box'  any square of the form
$$[t\log2, (t+1)\log2] \times [t\log2 (t+1)\log2],$$ for some integer $0 \leq t \leq n-1$.
\smallskip

First note that that if $(u^*, v^*)$ achieves the minimum in (\ref{gh1h2}), then it must be that $(u^*, v^*)$ is inside a diagonal box, and so is $(x - u^*, y-v^*)$. To see this consider for instance the case when $(u^*,v^*)$ lies `below' a diagonal box. In this case we can increase $v^*$ (till we hit the diagonal box) while keeping the value of $f_{H_2}(u^*,v^*)$ constant ($ = u^*$) and simultaneously decrease the value of  $f_{H_1}(x-u^*,y-v^*)$, thus decreasing the value of the sum. To be precise, suppose that 
$$k\log2 \leq x,y \leq (k+1)\log2,$$
$$m\log2 \leq u^* \leq \min(x,(m +1)\log2),$$
where $m \leq k$. Suppose also that
$$v^* < m\log2.$$
Then we have
$$f_{H_2}(u^*,v^*)=  f_{H_2}(u^*,m\log2) = u^*$$
and by monotonicity of $f_{H_1}$ we also have
$$f_{H_1}(x-u^*,y-v^*) \geq f_{H_1}(x-u^*,y-m\log2).$$
Note also that we have 
$$m\log2 \leq k\log2 \leq y$$ 
and also that 
$$m\log2 \leq u^* \leq \log|H_2|.$$ Thus $m\log2$ satisfies (\ref{condv}) and is a valid choice for $v$. This shows that  the optimal $(u^*, v^*)$ can be taken to lie in the diagonal box $[m\log2, (m+1)\log2] \times [m\log2, (m+1)\log2]$. Similar logic holds for when $(u^*, v^*)$ lies `to the left' of a diagonal box, or when $(x-u^*, y-v^*)$ lies `above' or `to the left' of a diagonal box.
\medskip

Our strategy will be as follows, we first use the above criteria on the optimal $(u^*,v^*)$ to restrict the domain of $(u,v)$ to a number of sub-rectangles of the diagonal boxes. We then use the induction hypothesis and reduce the problem of minimizing $ f_{H_2}(u,v) + f_{H_1}(x-u,y-v)$ to that of minimizing $f_{2^{l_2}}(u,v) + f_{2^{l_1}}(x-u,y-v)$. We examine the value of  $\min f_{2^{l_2}}(u,v) + f_{2^{l_1}}(x-u,y-v)$ over the rectangles, one rectangle at a time. The minimum over a single rectangle can be determined from the proof of the $\mathbb{Z}_{2^n}$ case, and it turns out to be $f_{2^n}(x,y)$ independent of which rectangle we choose. Thus the overall minimum also turns out to be $f_{2^n}(x,y)$ .
\medskip
 
Let $k\log2 \leq x,y \leq (k+1)\log2$. 

Let us write 
$$x = k\log2 + x',$$
$$y = k\log2 + y',$$
where $0 \leq x',y' \leq \log2$ and define the rectangles
\begin{align*}
R_0 &\triangleq [0,x'] \times [0,y']~,\\
R_1 &\triangleq [\log2, \log2+x'] \times [\log2, \log2+y']~,\\
&.\\
&.\\
R_{k} &\triangleq [k\log2, k\log2+x'] \times [k\log2, k\log2+y']~,\\
\end{align*}
and
\begin{align*}
S_1 &\triangleq [x', \log2] \times [y', \log2]~,\\
S_2 &\triangleq [\log2+x', 2\log2] \times [\log2+y', 2\log2]~,\\
&.\\
&.\\
S_{k} &\triangleq [[(k-1)\log2+x', k\log2] \times [(k-1)\log2+y', k\log2]~.\\
\end{align*}
We consider three separate cases $$k+1 \leq l_1,$$ $$l_1 +1 \leq k+1 \leq l_2,$$ $$l_2+1 \leq k+1 \leq n.$$ 
In the first case, the set of $(u,v)$ that satisfy (\ref{condu}),(\ref{condv}) and such that $(u,v)$ and $(x-u,y-v)$ both lie in diagonal boxes is $\left(\cup_{m=0}^k R_m\right) \cup \left(\cup_{m=1}^{k}S_m\right)$. In the second case it is  $\left(\cup_{m=k-l_1+1}^k R_m\right) \cup \left(\cup_{m=k-l_1+1}^{k}S_m\right)$, and in the third case it is  $\left(\cup_{m=l-l_1+1}^{l_2-1} R_m\right) \cup \left(\cup_{m=k-l_1+1}^{l_2}S_m\right)$.
\medskip

Fix $0 \leq m \leq k$ and consider
$$\min_{(u,v) \in R_m} \left( f_{H_2}(u,v) + f_{H_1}(x-u,y-v) \right)$$
assuming that we are in one of the three cases where all $(u,v) \in R_m$ satisfy equations (\ref{condu}),(\ref{condv}). Let us write 
$$u = m\log2 + u',$$
$$v = m\log2 + v',$$
where $0 \leq u' \leq x'$ and $0 \leq v' \leq y'$. By induction hypothesis we have 
$$f_{H_2}(u,v) = f_{2^{l_2}}(u,v) = m\log2 + f_2(u',v')$$
and
$$f_{H_1}(x-u,y-v) = f_{2^{l_1}}(x-u,y-v) = (k-m)\log2 + f_2(x'-u',y'-v')~.$$
Hence
\begin{align*}
\min_{(u,v) \in R_m} \left( f_{H_2}(u,v) + f_{H_1}(x-u,y-v) \right) &= k\log2 + \min_{u',v'} f_2(u',v') + f_2(x'-u',y'-v')\\
&\stackrel{(a)}{=} k\log2 + f_2(x',y')~.
\end{align*}
Here $(a)$ follows from the proof of the $\mathbb{Z}_4$ case. Note that this equals $f_{2^n}(x,y)$. \medskip

Now fix $1 \leq m \leq k$ and consider
$$\min_{(u,v) \in S_m} \left( f_{H_2}(u,v) + f_{H_1}(x-u,y-v) \right)$$
assuming that we are in one of the three cases where all $(u,v) \in S_m$ satisfy equations (\ref{condu}),(\ref{condv}). Note that this is equivalent to requiring that we are in one of the cases where all $(x-u,y-v)$ for $(u,v)$ satisfy (\ref{condua}),(\ref{condva}). Let us write
$$u = (m-1)\log2 + u',$$
$$v = (m-1)\log2 + v',$$
 where $x' \leq u' \leq \log2$ and $y' \leq v' \leq \log2$. By inductive hypothesis we have
$$f_{H_2}(u,v) = f_{2^{l_2}}(u,v) = (m-1)\log2 + f_2(u',v')~.$$
Further, since
$$x-u = (k-m)\log2 + \log2 + x' - u',$$
$$y-v = (k-m)\log2 + \log2 + y' - v',$$
with $x' \leq \log2 + x' - u' \leq \log2$ and $y' \leq \log2 + y' - v' \leq \log2$, by inductive hypothesis we have
$$f_{H_1}(x-u,y-v) = f_{2^{l_1}}(x-u,y-v) = (k-m)\log2 + f_2(\log2 + x' - u', \log2 + y' - v').$$
Hence we have
\begin{align*}
&\min_{(u,v) \in S_m} \left( f_{H_2}(u,v) + f_{H_1}(x-u,y-v) \right)\\
 &= (k-1)\log2 + \min_{x' \leq u' \leq \log2, 0 \leq v' \leq \log2} f_2(u',v') +f_2(\log2 + x' - u', \log2 + y' - v')\\
&\stackrel{(a)}{=} (k-1)\log2 + f_4(\log2+x',\log2+y')\\
&\stackrel{(b)}{=} k\log2 + f_2(x',y')~,
\end{align*}
where $(a),(b)$ again follows from the proof of the $\mathbb{Z}_4$ case. Note that that this equals $f_{2^n}(x,y)$. This completes the proof of Claim \ref{abelianstep2}, and thus of Theorem \ref{abelian}.
\end{proof}
\end{proof}

\section{Extensions}
In this section we will prove some extensions of the earlier results that seem to be of potential interest.

\subsection{Scalar and Vector MGL}

\begin{claim}\label{scalarMGL}

Let $X$, $Y$ and $Z$ be random variables taking values in an abelian group $G$ of order $2^n$, and let $U$ be an arbitrary random variable. Suppose $Z$ is independent of $(U,X)$ and $Y = X + Z$ where the addition is understood to be the group addition. Then
$$H(Y|U) \geq f_G(H(X|U),H(Z))$$
\end{claim}
\begin{remark}
In the case of binary random variables $X$, $Y$, and $Z$, where $Z \sim Bern(p)$, $U$ is an arbitrary random variable, and $Z$ is independent of $(U,X)$, one has the scalar MGL given by 
$$H(Y|U) \geq h(h^{-1}(H(X|U)) \star p)~.$$
Thus, Claim \ref{scalarMGL} can be thought of as the generalization of this scalar MGL for random variables taking 
values in an abelian group of order $2^n$.
\end{remark}

\begin{proof}[Proof of Claim \ref{scalarMGL}]
\begin{align}
H(Y|U) &= \sum_{u} P(U=u)H(Y|U=u)\\
&= \sum_{u}P(U=u)H(X+Z|U=u)\\
&\geq \sum_{u}P(U=u)f_G(H(X|U=u),H(Z)) \label{scalarMGL3}\\
& \geq f_G \left( \sum_{u} P(U=u)H(X|U=u), H(Z)\right) \label{scalarMGL4}\\
&= f_G(H(X|U),H(Z))~,
\end{align}
where (\ref{scalarMGL3}) follows from the definition of $f_G$ and (\ref{scalarMGL4}) follows from the convexity of $f_G(x,y)$
in $x$ for fixed $y$.
\end{proof}

\begin{claim}\label{vectorMGL}
Let $X^k$ be a random vector each of whose coordinates takes values in an abelian group $G$ of order $2^n$, and let $U$ be an arbitrary random variable. If $Z^k$ is a vector of independent and identically distributed $G$-valued random variables, each distributed according to $p_Z$, and $Z^k$ is independent of $(X^k, U)$, with $Y^k = X^k + Z^k$ then
$$\frac{H(Y^k|U)}{k} \geq f_G \left( \frac{H(X^k|U)}{k}, H(Z) \right)~. $$
\end{claim}
\begin{remark}
Claim \ref{vectorMGL} in the case of binary random variables where $Z^k$ is a vector of i.i.d. random variables having distribution Bern($p$), is given by
$$\frac{H(Y^k|U)}{k} \geq h \left(h^{-1}\left( \frac{H(X^k|U)}{k}\right) \star p \right)~, $$
and is known to be true.
Thus, Claim \ref{vectorMGL} can be thought of as the vector MGL for random variables taking values in an abelian group of order $2^n$.
\end{remark}

\begin{proof}[Proof of Claim \ref{vectorMGL} ]

\begin{align}
\frac{H(Y^k|U)}{k} &= \sum_{i=1}^{k} \frac{H(Y_i|U, Y^{i-1})}{k}\\
& \geq \sum_{i=1}^{k} \frac{H(Y_i|U, Y^{i-1},X^{i-1})}{k} \label{vectorMGL2}\\
&= \sum_{i=1}^{k} \frac{H(Y_i|U, X^{i-1})}{k} \label{vectorMGL3} \\
& \geq \sum_{i=1}^{k} \frac{ f_G(H(X_i | U, X^{i-1}), H(Z))}{k} \label{vectorMGL4}\\
& \geq  f_G \left( \sum_{i=1}^{k}\frac{H(X_i | U, X^{i-1})}{k}, H(Z)\right) \label{vectorMGL5}\\
&= f_G \left( \frac{H(X^k | U)}{k}, H(Z) \right)~.
\end{align}
Here, (\ref{vectorMGL2}) is because conditioning reduces entropy, (\ref{vectorMGL3}) is because the channel from $X^k$ to $Y^k$ is a DMC, (\ref{vectorMGL4}) follows from the scalar MGL, (\ref{vectorMGL5}) is because of the convexity of $f_G(x,y)$ in $x$ for fixed $y$.
\end{proof}

\subsection{The minimum entropy of a sum of $k \ge 2$ independent $G$-valued random variables with fixed entropies}

Consider an abelian group $G$ of order $2^n$ and $k \ge 2$ independent random variables $X_1,X_2,...,X_k$ taking values in $G$. We define the function 
\begin{equation}
f_{G,k}(x_1,x_2,...,x_k) := \min_{H(X_i) = x_i, 1\leq i \leq k}H(X_1+X_2+...+X_k)~.
\end{equation}
The function $f_{G,1}$ is the identity function, whereas our earlier function $f_G$ can be thought of as $f_{G,2}$. 

We divide the interval $[0,n\log2]$  into $n$ blocks of size $\log2$, namely $[i\log2, (i+1)\log2]$, where  $0\leq i \leq n-1$. We bin $x_1,x_2,...,x_k$ into these $n$ bins and consider the largest $m$ such that $m\log2 \leq x_l \leq (m+1)\log2$ for some
$1 \le l \le k$. Let the contents of this bin be  $x^1,x^2,...,x^r$ where $r \leq k$. Call the corresponding random variables $X^1,X^2,...,X^r$. We claim the following:
\begin{claim}\label{krvsclaim}
$$\min_{H(X^i) = x_i, 1\leq i \leq r}H(X^1+X^2+...+X^r) = f_{G,2}(x^1, f_{G,2}(x^2,(...(f_{G,2}(x^{r-1},x^r)))..))~.$$  
\end{claim}
\begin{proof}[Proof of Claim \ref{krvsclaim}]
Note that $$H(X^r +X^{r-1}) \geq f_{G,2}(x^r,x^{r-1}),$$ by definition of $f_{G,2}$. Now by monotonicity of $f_{G,2}$, we also have 
$$H(X^{r-2} + X^{r-1} + X^{r}) \geq f_{G,2}(x^{r-2}, H(X^{r-1}+X^r)) \geq f_{G,2}(x^{r-2}, f_{G,2}(x^{r-1},x^r)).$$
Continuing in a similar fashion, we get
$$H(X^1+X^2+...+X^r) \geq f_{G,2}(x^1, f_{G,2}(x^2,(...(f_{G,2}(x^{r-1},x^r)))..))$$
for whatever choice of distributions of $X^1,X^2,...,X^r$. This gives us the lower bound
\begin{equation}
\min_{H(X^i) = x^i, 1\leq i \leq r}H(X^1+X^2+...+X^r) \geq f_{G,2}(x^1, f_{G,2}(x^2,(...(f_{G,2}(x^{r-1},x^r)))..)).
\end{equation}
Now consider a group $H_1$ of order $2^{m+1}$ and its subgroup $H_2$ of order $2^m$. Let $X^1,X^2,...X^{r}$  have distributions supported on $H_1$ such that they take constant values on the cosets  $H_1/H_2$ and satisfy $H(X^i) = x^i$ for $1 \leq i \leq r$. Let these distributions be $p_{X^1},p_{X^2},...,p_{X^r}$. For this choice of distributions, we have 
$$H(X^r +X^{r-1}) = f_{G,2}(x^r,x^{r-1}),$$
since these distributions achieve equality for $f_{G,2}$. We also have
$$H(X^{r-2} + X^{r-1} + X^{r}) = f_{G,2}(x^{r-2}, f_{G,2}(x^{r-1},x^r)),$$
as $p_{X^{r-2}}$ and $p_{X^{r-1}} \circledast_G p_{X^r}$ are equality achieving distributions for $f_{G,2}$. 
Continuing similarly, we see that the lower bound is achieved, thus proving the claim.
\end{proof}

\begin{claim}\label{krvsclaim2}
$$f_{G,k}(x_1,...,x_k)= f_{G,2}(x^1, f_{G,2}(x^2,(...(f_{G,2}(x^{r-1},x^r)))..))~.$$
\end{claim}
\begin{proof}[Proof of Claim \ref{krvsclaim2}]
Note that since $k \geq r$, and by Claim \ref{krvsclaim}, we have the lower bound
\begin{equation}\label{lbkrvs}
f_{G,k}(x_1,...,x_k) \geq f_{G,r}(x^1,x^2,...,x^r) = f_{G,2}(x^1, f_{G,2}(x^2,(...(f_{G,2}(x^{r-1},x^r)))..)).
\end{equation}
We'll show that this lower bound is attained. Consider a group $H_1$ of size $2^{m+1}$ and its subgroup $H_2$ of size $2^{m}$. Define distributions of  $X^1,X^2,...X^{r}$ supported on $H_1$ such that they take constant values on the cosets  $H_1/H_2$ and satisfy $H(X^i) = x^i$ for $1 \leq i \leq r$. Let the remaining random variables take arbitrary distributions supported on either of the two cosets of $H_2$ in $H_1$, and such that they satisfy the entropy constraints. 
It is easily checked that 
\begin{equation}\label{krvsconv}
p_{X_1} \circledast_G p_{X_2} \circledast_G ... \circledast_G p_{X_k} = p_{X^1} \circledast_G p_{X^2} \circledast_G ... \circledast_G p_{X^r},
\end{equation}
giving us
$$H(X_1+X_2+...+X_k) = H(X_1+X_2+...+X_r) =  f_{G,2}(x_1, f_{G,2}(x_2,(...(f_{G,2}(x_{r-1},x_r)))..))$$
where the second equality follows from Claim \ref{krvsclaim}. By the definition of $f_{G,k}$, this gives us
\begin{equation}\label{ubkrvs}
f_{G,k}(x_1,...,x_k) \leq f_{G,2}(x_1, f_{G,2}(x_2,(...(f_{G,2}(x_{r-1},x_r)))..))~.
\end{equation}
\medskip
Equations (\ref{lbkrvs}) and (\ref{ubkrvs}) prove Claim \ref{krvsclaim2}. 
\end{proof}

\begin{theorem}\label{krvs}
Given any $x_1,x_2,...,x_k$ we have 
$$f_{G,k}(x_1,x_2,...,x_k) = f_{G,2}(x_1, f_{G,2}(x_2,(...(f_{G,2}(x_{k-1},x_k)))..))~.$$
\end{theorem}
\begin{proof}[Proof of Theorem \ref{krvs}]
Let $r$ be as before and let $x_{i_1}, x_{i_2},...,x_{i_r}$ be those $x_i$'s which land in the largest bin, $[m\log2, (m+1)\log2]$. Let $1 \leq i_1 < i_2 < ... < i_r \leq k$. It is easy to see that 
$$f_{G,2}(x_{i_{r-1}}, f_{G,2}(x_{i_{r-1}+1}, f_{G,2}( ..., x_k))..)) = f_{G,2}((x_{i_{r-1}}, x_{i_r})).$$
Continuing in a similar manner, we get
$$f_{G,2}(x_1, f_{G,2}(x_2,(...(f_{G,2}(x_{k-1},x_k)))..)) =  f_{G,2}(x_{i_1}, f_{G,2}(x_{i_2},(...(f_{G,2}(x_{i_{r-1}},x_{i_r})))..))$$
which by Claim \ref{krvsclaim2} is $f_{G,k}(x_1,x_2,...,x_k)$ thus proving Theorem \ref{krvs}.
\end{proof}

\begin{corollary}\label{corkrvs}
$f_{G,k}(x_1,x_2,...,x_k)$ is convex in each variable, when the remaining are kept fixed.
\end{corollary}
\begin{proof}[Proof of Corollary \ref{corkrvs}]
Without loss of generality, consider $x_k$ as varying and the remaining variables fixed. As before, let the largest bin in which atleast one $x_i$ is present be $[m\log2,(m+1)\log2]$. Now as long as $x_k < m\log2$, $$f_{G,k}(x_1,...,x_k) = f_{G,k-1}(x_1,x_2,..,x_{k-1})$$ which is constant as a function of $x_k$. For $m\log2 \leq x_k \leq (m+1)\log2$, we have that $$f_{G,k}(x_1,x_2,...,x_k) = f_{G,2}(x_k,f_{G,k-1}(x_1,x_2,....,x_{k-1}))$$ which is convex in $x_k$ by MGL. For $x_k > (m+1)\log2$, $$f_{G,k}(x_1,x_2,...,x_k) = x_k.$$ Now the convexity easily follows from MGL and Claim \ref{dfdx<1}.
\end{proof}

\section{Acknowledgements}
Research support from the ARO MURI grant W911NF-08-1-0233, ``Tools for the Analysis and Design of Complex Multi-scale Network", from the NSF grant CNS-0910702, from the NSF Science \& Technology Center grant CCF-0939370, ``Science of Information", from Marvell Semiconductor Inc., and from the U.C. Discovery program is gratefully acknowledged.

\bibliographystyle{ieeetr}	
\bibliography{myrefs}		

\begin{thebibliography}{10}

\bibitem{shannon1948}
C.~Shannon, ``A mathematical theory of communications, {I} and {II},'' {\em
  Bell Syst. Tech. J}, vol.~27, pp.~379--423, 1948.

\bibitem{stam1959}
A.~Stam, ``Some inequalities satisfied by the quantities of information of
  {F}isher and {S}hannon,'' {\em Information and Control}, vol.~2, no.~2,
  pp.~101--112, 1959.

\bibitem{blachman1965}
N.~Blachman, ``The convolution inequality for entropy powers,'' {\em
  Information Theory, IEEE Transactions on}, vol.~11, no.~2, pp.~267--271,
  1965.

\bibitem{lieb1978}
E.~Lieb, ``Proof of an entropy conjecture of {W}ehrl,'' {\em Communications in
  Mathematical Physics}, vol.~62, no.~1, pp.~35--41, 1978.

\bibitem{verdu2006}
S.~Verd{\'u} and D.~Guo, ``A simple proof of the entropy-power inequality,''
  {\em Information Theory, IEEE Transactions on}, vol.~52, no.~5,
  pp.~2165--2166, 2006.

\bibitem{rioul2011}
O.~Rioul, ``Information theoretic proofs of entropy power inequalities,'' {\em
  Information Theory, IEEE Transactions on}, vol.~57, no.~1, pp.~33--55, 2011.

\bibitem{bergmans1973}
P.~Bergmans, ``Random coding theorem for broadcast channels with degraded
  components,'' {\em Information Theory, IEEE Transactions on}, vol.~19, no.~2,
  pp.~197--207, 1973.

\bibitem{leung1978}
S.~Leung-Yan-Cheong and M.~Hellman, ``The {G}aussian wire-tap channel,'' {\em
  Information Theory, IEEE Transactions on}, vol.~24, no.~4, pp.~451--456,
  1978.

\bibitem{ozarow1980}
L.~Ozarow, ``On a source-coding problem with two channels and three
  receivers,'' {\em Bell Syst. Tech. J}, vol.~59, no.~10, pp.~1909--1921, 1980.

\bibitem{oohama1998}
Y.~Oohama, ``The rate-distortion function for the quadratic {G}aussian {CEO}
  problem,'' {\em Information Theory, IEEE Transactions on}, vol.~44, no.~3,
  pp.~1057--1070, 1998.

\bibitem{weingarten2006mimo}
H.~Weingarten, Y.~Steinberg, and S.~Shamai, ``The capacity region of the
  {G}aussian multiple-input multiple-output broadcast channel,'' {\em
  Information Theory, IEEE Transactions on}, vol.~52, no.~9, pp.~3936--3964,
  2006.

\bibitem{costa1985}
M.~Costa, ``A new entropy power inequality,'' {\em Information Theory, IEEE
  Transactions on}, vol.~31, no.~6, pp.~751--760, 1985.

\bibitem{dembo1989}
A.~Dembo, ``Simple proof of the concavity of the entropy power with respect to
  added {G}aussian noise,'' {\em Information Theory, IEEE Transactions on},
  vol.~35, no.~4, pp.~887--888, 1989.

\bibitem{villani2000}
C.~Villani, ``A short proof of the concavity of entropy power,'' {\em IEEE
  Transactions on Information Theory}, vol.~46, no.~4, pp.~1695--1696, 2000.

\bibitem{zamir1993generalization}
R.~Zamir and M.~Feder, ``A generalization of the entropy power inequality with
  applications,'' {\em Information Theory, IEEE Transactions on}, vol.~39,
  no.~5, pp.~1723--1728, 1993.

\bibitem{viswanathliu2007}
T.~Liu and P.~Viswanath, ``An extremal inequality motivated by multiterminal
  information-theoretic problems,'' {\em Information Theory, IEEE Transactions
  on}, vol.~53, no.~5, pp.~1839--1851, 2007.

\bibitem{liu2010vector}
R.~Liu, T.~Liu, H.~Poor, and S.~Shamai, ``A vector generalization of {C}osta's
  entropy-power inequality with applications,'' {\em Information Theory, IEEE
  Transactions on}, vol.~56, no.~4, pp.~1865--1879, 2010.

\bibitem{artstein2004}
S.~Artstein, K.~Ball, F.~Barthe, and A.~Naor, ``Solution of {S}hannon's problem
  on the monotonicity of entropy,'' {\em Journal of the American Mathematical
  Society}, vol.~17, no.~4, pp.~975--982, 2004.

\bibitem{tulino2006}
A.~Tulino and S.~Verd{\'u}, ``Monotonic decrease of the non-{G}aussianness of
  the sum of independent random variables: A simple proof,'' {\em Information
  Theory, IEEE Transactions on}, vol.~52, no.~9, pp.~4295--4297, 2006.

\bibitem{madiman2007generalized}
M.~Madiman and A.~Barron, ``Generalized entropy power inequalities and
  monotonicity properties of information,'' {\em Information Theory, IEEE
  Transactions on}, vol.~53, no.~7, pp.~2317--2329, 2007.

\bibitem{wyner1973part1}
A.~Wyner and J.~Ziv, ``A theorem on the entropy of certain binary sequences and
  applications--{I},'' {\em Information Theory, IEEE Transactions on}, vol.~19,
  no.~6, pp.~769--772, 1973.

\bibitem{wyner1973part2}
A.~Wyner, ``A theorem on the entropy of certain binary sequences and
  applications--{II},'' {\em Information Theory, IEEE Transactions on},
  vol.~19, no.~6, pp.~772--777, 1973.

\bibitem{witsenhausen1974}
H.~Witsenhausen, ``Entropy inequalities for discrete channels,'' {\em
  Information Theory, IEEE Transactions on}, vol.~20, no.~5, pp.~610--616,
  1974.

\bibitem{shamai1990}
S.~Shamai and A.~Wyner, ``A binary analog to the entropy-power inequality,''
  {\em Information Theory, IEEE Transactions on}, vol.~36, no.~6,
  pp.~1428--1430, 1990.

\bibitem{harremoes2003epi}
P.~Harremoes, C.~Vignat, {\em et~al.}, ``An entropy power inequality for the
  binomial family,'' {\em JIPAM. J. Inequal. Pure Appl. Math}, vol.~4, no.~5,
  2003.

\bibitem{sharma2011epi}
N.~Sharma, S.~Das, and S.~Muthukrishnan, ``Entropy power inequality for a
  family of discrete random variables,'' in {\em Information Theory Proceedings
  (ISIT), 2011 IEEE International Symposium on}, pp.~1945--1949, IEEE, 2011.

\bibitem{johnson2010monotonicity}
O.~Johnson and Y.~Yu, ``Monotonicity, thinning, and discrete versions of the
  entropy power inequality,'' {\em Information Theory, IEEE Transactions on},
  vol.~56, no.~11, pp.~5387--5395, 2010.

\bibitem{tao2010sumset}
T.~Tao, ``Sumset and inverse sumset theory for shannon entropy,'' {\em
  Combinatorics, Probability \& Computing}, vol.~19, no.~4, pp.~603--639, 2010.

\bibitem{taovu2006}
T.~Tao and V.~Vu, {\em Additive combinatorics}, vol.~105.
\newblock Cambridge Univ Pr, 2006.

\bibitem{kontoyiannis2012sumset}
I.~Kontoyiannis and M.~Madiman, ``Sumset and inverse sumset inequalities for
  differential entropy and mutual information,'' {\em Arxiv preprint
  arXiv:1206.0489}, 2012.

\bibitem{ahlswede1974}
R.~Ahlswede and J.~K{\"o}rner, ``On the connection between the entropies of
  input and output distributions of discrete memoryless channels,'' in {\em
  Proceedings of the Fifth Conference on Probability Theory, Brasov},
  pp.~13--22, 1974.

\bibitem{apostol1966}
T.~Apostol, ``Calculus. one-variable calculus with an introduction to linear
  algebra,'' 1966.

\bibitem{jacobson1985}
N.~Jacobson, ``Basic algebra, volume i,'' 1985.

\end{thebibliography}
\appendix
\section{Proof of Lemma \ref{dfdx}} \label{lemma1}
We'll first compute $\frac{\partial f}{\partial x}$. 
Let $x = h(p)$ and $y = h(q)$ with $0 \le p, q \le \frac{1}{2}$,
so $f(x,y) = h(p \star q)$.

\begin{align}
\frac{\partial f}{\partial x} =& \frac{\partial f}{\partial p}  \frac{\partial p}{\partial x}\\
=&(1-2q)\log\left(\frac{1-p \star q}{p \star q}\right) \times \frac{1}{\log\left(\frac{1-p}{p}\right)}~.
\end{align}
Notice that as $x$ moves along a line with slope $\theta >  0$, $p$ and $q$ both strictly increase and consequently the function $(1-2q)$ strictly decreases. Therefore to show $\frac{\partial f}{\partial x}$ strictly decreases, it is enough to show that 
$\log\left(\frac{1-p \star q}{p \star q}\right) \times \frac{1}{\log\left(\frac{1-p}{p}\right)}$ monotonically decreases along the line. Now let us compute the directional derivative of $g(x,y) = \frac{\log\left(\frac{1-p \star q}{p \star q}\right)}  {\log\left(\frac{1-p}{p}\right)}$ at a point $(x,y) \in (0,1) \times (0,1)$, as we move in a direction $(1,\theta)$.
\begin{align}
\frac{\partial g}{\partial x} + \theta\frac{\partial g}{\partial y} &= \frac{\partial g}{\partial p}\frac{\partial p }{\partial x} +\theta \frac{\partial g}{\partial q}\frac{\partial q}{\partial y}\\
&= \frac{\log\left(\frac{1-p}{p}\right)\left(-\frac{1}{1-p \star q} - \frac{1}{p \star q}\right)(1-2q) - \log\left(\frac{1-p \star q}{p \star q}\right)\left(-\frac{1}{1-p}-\frac{1}{p}\right)}{\left(\log\left(\frac{1-p}{p}\right)\right)^2} \times \frac{1}{\log\left(\frac{1-p}{p}\right)}\\ 
&+ \theta \frac{\left(-\frac{1}{1-p \star q} - \frac{1}{p \star q}\right)(1-2p)}{\log\left(\frac{1-p}{p}\right)} \times \frac{1}{\log\left(\frac{1-q}{q}\right)}  \nonumber\\
\label{dgdx}
&= \frac{-(1-2q)}{(p \star q)(1- p\star q)\left(\log\left(\frac{1-p}{p}\right)\right)^2} + \frac{\log\left(\frac{1-p \star q}{p \star q}\right)}{p(1-p)\left(\log\left(\frac{1-p}{p}\right)\right)^3}\\ 
&- \theta \frac{(1-2p)}{(p \star q)(1- p\star q)\log\left(\frac{1-p}{p}\right)\log\left(\frac{1-q}{q}\right)}~. \nonumber
\end{align}
Now we choose $\theta = \frac{h(q)}{h(p)}$. We want to show that with this choice of $\theta$, $(\ref{dgdx}) \leq 0$, since this would mean $g(x,y)$ decreases as we move in the desired direction. Thus we see that it suffices to show
\begin{equation}
 \frac{-(1-2q)}{(p \star q)(1- p\star q)\left(\log\left(\frac{1-p}{p}\right)\right)^2} + \frac{\log\left(\frac{1-p \star q}{p \star q}\right)}{p(1-p)\left(\log\left(\frac{1-p}{p}\right)\right)^3} -   \frac{h(q)(1-2p)}{h(p)(p \star q)(1- p\star q)\log\left(\frac{1-p}{p}\right)\log\left(\frac{1-q}{q}\right)} \stackrel{?}{\leq} 0~.
\end{equation}
Note that  since $(x,y)$ lies in the interior, $0< p,q < \frac{1}{2}$. Multiplying throughout by $\log \left(\frac{1-p}{p}\right)$ we need to show
\begin{equation}
 \frac{-(1-2q)}{(p \star q)(1- p\star q)\log\left(\frac{1-p}{p}\right)} + \frac{\log\left(\frac{1-p \star q}{p \star q}\right)}{p(1-p)\left(\log\left(\frac{1-p}{p}\right)\right)^2} -   \frac{h(q)(1-2p)}{h(p)(p \star q)(1- p\star q)\log\left(\frac{1-q}{q}\right)} \stackrel{?}{\leq} 0~.
\end{equation}
Taking the negative terms on the other side, we need to show
\begin{equation}
 \frac{\log\left(\frac{1-p \star q}{p \star q}\right)}{p(1-p)\left(\log\left(\frac{1-p}{p}\right)\right)^2}   \stackrel{?}{\leq} \frac{(1-2q)}{(p \star q)(1- p\star q)\log\left(\frac{1-p}{p}\right)} + \frac{h(q)(1-2p)}{h(p)(p \star q)(1- p\star q)\log\left(\frac{1-q}{q}\right)}~.
\end{equation}
Multiplying by $(p \star q)(1- p \star q)$ on both sides, we need to show
\begin{equation}
 \frac{(p \star q)(1- p\star q)\log\left(\frac{1-p \star q}{p \star q}\right)}{p(1-p)\left(\log\left(\frac{1-p}{p}\right)\right)^2}   \stackrel{?}{\leq} \frac{(1-2q)}{\log\left(\frac{1-p}{p}\right)} + \frac{h(q)(1-2p)}{h(p)\log\left(\frac{1-q}{q}\right)}~.
\end{equation}
Now multiplying both sides by $p(1-p)\left(\log\left(\frac{1-p}{p}\right)\right)^2$, we need to show
\begin{equation}\label{pstarq}
(p \star q)(1- p\star q)\log\left(\frac{1-p \star q}{p \star q}\right)  \stackrel{?}{\leq} (1-2q)p(1-p)\log\left(\frac{1-p}{p}\right) + \frac{p(1-p)(1-2p)\left(\log\left(\frac{1-p}{p}\right)\right)^2 h(q)}{h(p)\log\left(\frac{1-q}{q}\right)}~.
\end{equation}
We'll now analyse (\ref{pstarq}) by keeping the left side fixed and finding the minimum of the right side. Let $p \star q = k$. Note that $p \leq k$ and $q = \frac{k-p}{1-2p}$. Observe that when $p=k$, $q=0$ and the first term on the right side equals the left side, whereas the second term is 0 (it is easy to see that $\frac{h(q)}{\log\left(\frac{1-q}{q}\right)} \to 0$ as $q \to 0$). Thus, it will be sufficient to show that the right hand side is a decreasing function of $p$ if  $p \star q$ is fixed. Substitute $q$ in the first term to get
\begin{equation}\label{pstarq2}
\frac{(1-2k)}{(1-2p)}p(1-p)\log\left(\frac{1-p}{p}\right) + \frac{p(1-p)(1-2p)\left(\log\left(\frac{1-p}{p}\right)\right)^2 h(q)}{h(p)\log\left(\frac{1-q}{q}\right)}~.
\end{equation}
Showing (\ref{pstarq2}) decreases in $p$ for a fixed $k$ is equivalent to showing $A(p,k)$ decreases in $p$ where $A$ is given by
\begin{equation}
A(p,k) = \frac{1}{(1-2p)}p(1-p)\log\left(\frac{1-p}{p}\right) + \frac{p(1-p)(1-2p)\left(\log\left(\frac{1-p}{p}\right)\right)^2 h(q)}{h(p)\log\left(\frac{1-q}{q}\right)(1-2k)}~.
\end{equation}
For ease of notation, rename the following functions
 $$B(p,k) := \frac{\partial A}{\partial p}~,$$
 $$M(p) :=  \frac{p(1-p)}{(1-2p)}\log\left(\frac{1-p}{p}\right)~,$$
 $$N(p) := \frac{p(1-p)(1-2p)\left(\log\left(\frac{1-p}{p}\right)\right)^2 }{h(p)}~,$$
 $$L(q) := \frac{h(q)}{\log\left(\frac{1-q}{q}\right)}~.$$
 So we have (note that in the equation below $q$ is thought of as a function of $p$ and $k$)
 \begin{equation}
 A(p,k) = M(p) + \frac{N(p)L(q)}{1-2k}~.
 \end{equation}
 Differentiating w.r.t $p$, we get
 \begin{align}
 B(p,k) &= M'(p) + \frac{N'(p)L(q)}{1-2k} + \frac{N(p)\frac{d L(q)}{dq}\frac{dq}{dp}}{1-2k}\\
 \label{pstarq3}
 &= M'(p) + N'(p) \frac{L(q)}{1-2k} + \frac{N(p)}{(1-2p)^2}\left(-\frac{d L(q)}{dq}\right)
 \end{align}
 where (\ref{pstarq3}) is got by $\frac{dq}{dp} = -\frac{1-2k}{(1-2p)^2}$.
 We want to show that $B(p,k) \leq 0$ for all valid pairs $(p,k)$ (a pair is valid if $0 < p \leq k$). It is therefore sufficient to show that $\max_{k \geq p} B(p,k) \leq 0$. We now make two claims.
 \begin{claim}\label{np}
 $N'(p) \leq 0$ i.e $N(p)$ is a decreasing function of $p$, as $p$ goes from $0$ to $\frac{1}{2}$.
 \end{claim}
 \begin{figure}[h!]
  
  \centering
    \includegraphics[scale = 0.75]{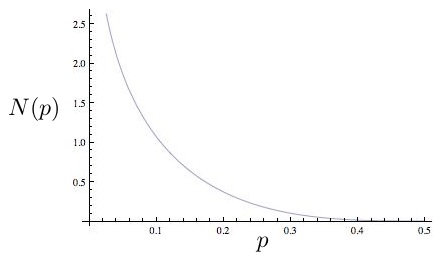}
    \caption{Plot of $N(p)$}
\end{figure}

 \begin{claim}\label{lp}
 $L(x) = \frac{h(x)}{\log\left(\frac{1-x}{x}\right)}$ is an increasing function as $x$ goes from $0$ to $\frac{1}{2}$, and $L'(x)$ is minimum at $x = 0$.
 \end{claim}
 
 \begin{figure}[h!]
  
  \centering
    \includegraphics[scale = 0.75]{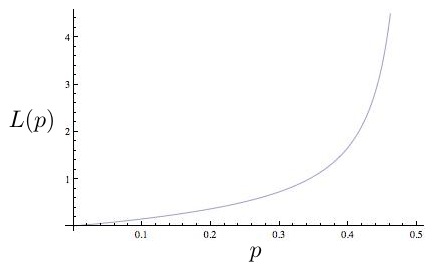}
    \caption{Plot of $L(p)$}
\end{figure}

 \smallskip
 
 Suppose we did term-by-term maximisation of $B(p,k)$ as $k$ varies. $M'(p)$ does not depend on $k$, so we don't need to care about it. Now for the second term, since $N'(p) \leq 0$ (by Claim \ref{np}), to maximise $B$  we need to minimise $\frac{L(q)}{1-2k}$ as a function of $k$. Now we note that as $k \uparrow$, $q \uparrow$ and by claim 2 we get $L(q) \uparrow$. Also clearly as $k \uparrow$, $\frac{1}{1-2k} \uparrow$. Thus $\frac{L(q)}{1-2k}$ increases in $k$, and to minimise it, the best choice of $k$ is the minimum possible $k$, which is $p$. For the third term, because of the minus sign, we need to minimise $\frac{dL(q)}{dq}$. By Claim \ref{lp}, we see that this happens when $q = 0$ which happens when $k$ equals $p$. Thus, the above discussion leads us to conclude that $\arg\max_{k \geq p} B(p,k) = p$. It therefore suffices to show that 
 \begin{equation}\label{bpp}
 B(p,p) \stackrel{?}{\leq} 0 \mbox{      ~~~~~for all } 0 < p < \frac{1}{2}~.
 \end{equation}
 \medskip
 
 Having motivated the claims, we'll now prove them.
 \begin{proof}[Proof of Claim \ref{np}]
 Let's recall $N(p)$
  $$N(p) = \frac{p(1-p)(1-2p)\left(\log\left(\frac{1-p}{p}\right)\right)^2 }{h(p)}~.$$
Since $(1-p)$, $(1-2p)$ and $\log\left(\frac{1-p}{p}\right)$ are decreasing functions of $p$, we conclude that it suffices to prove $$\hat{N}(p) := \frac{p\log\left(\frac{1-p}{p}\right)}{h(p)}$$ decreases in $p$. Differentiating $\hat{N}$ and simplifying, we get that it suffices to show
\begin{equation}
\tilde{N}(p) := h(p)\log\left(\frac{1-p}{p}\right) - \frac{h(p)}{(1-p)} -p\left(\log\left(\frac{1-p}{p}\right)\right)^2 \stackrel{?}{\leq} 0~.
\end{equation}
Now as $p \to 0$, $\tilde{N}$ tends to $0$. Thus, to show that it is negative we'll show that $\tilde{N}' \leq 0$. Differentiating again, and simplifying we get that it suffices to show 
\begin{equation}\label{nprime}
h(p)\left(-\frac{1}{p} - \frac{1}{1-p}\right) + \log\left(\frac{1-p}{p}\right) \stackrel{?}{\leq} 0~.
\end{equation}
Now we expand $h(p) = -p\log(p)-(1-p)\log(1-p)$ and simplify $(\ref{nprime})$ to get 
\begin{align}
&\left(p\log(p)+(1-p)\log(1-p)\right)\left(\frac{1}{p} + \frac{1}{1-p}\right) + \log\left(\frac{1-p}{p}\right)   \stackrel{?}{\leq} 0\\
&\Leftrightarrow \bcancel{\log(p)} + \log(1-p) + \frac{p}{1-p}\log(p) + \frac{1-p}{p}\log(1-p) + \log(1-p) - \bcancel{\log(p)} \stackrel{?}{\leq} 0\\
&\Leftrightarrow 2\log(1-p) + \frac{p}{1-p}\log(p) + \frac{1-p}{p}\log(1-p)   \stackrel{?}{\leq} 0~,
\end{align}
which is immediate since $0 < p, 1-p < 1 $. This proves Claim \ref{np}.
 \end{proof}

\medskip
\begin{proof}[Proof of Claim \ref{lp}]
Recalling $L(x)$,
$$L(x) := \frac{h(x)}{\log\left(\frac{1-x}{x}\right)}~.$$ 
Differentiating, 
\begin{equation}
L'(x) = 1 + \frac{\frac{h(x)}{x(1-x)}}{\left(\log\left(\frac{1-x}{x}\right)\right)^2} \geq 1 > 0~.
\end{equation} 
Thus $L(x)$ is clearly an increasing function. To show that $L'(x)$ is minimum at $x=0$, we'll show that $L'(0) = 1$.
\begin{align*}
\lim_{x \to 0}\frac{h(x)}{x(1-x)\left(\log\left(\frac{1-x}{x}\right)\right)^2}  &= \lim_{x \to 0}\frac{h(x)}{x\left(\log\left(\frac{1-x}{x}\right)\right)^2}\\
&= \lim_{x \to 0}\frac{\log\left(\frac{1-x}{x}\right)}{\left(\log\left(\frac{1-x}{x}\right)\right)^2 - \frac{2}{1-x}\log\left(\frac{1-x}{x}\right)}\\
&= \lim_{x \to 0}\frac{1}{\log\left(\frac{1-x}{x}\right)-2}\\
&=0~.
\end{align*}
This proves claim \ref{lp}. 

\end{proof}

\medskip

Coming back to $(\ref{bpp})$,
 \begin{equation}
 B(p,p) = M'(p) + N'(p)\frac{L(0)}{1-2p} - \frac{N(p)}{(1-2p)^2}\frac{d L(q)}{dq}\bigg |_{q=0}~.
 \end{equation}
 $$L(0) = \lim_{q \to 0} \frac{H(q)}{\log\left(\frac{1-q}{q}\right)} = \frac{0}{\infty} = 0~.$$
 By the proof on Claim \ref{lp}, we also know that 
$$ \lim_{x \to 0} \frac{dL(x)}{dx} = 1~. $$
Using this, we get 
 \begin{equation}
 B(p,p) = M'(p)  - \frac{N(p)}{(1-2p)^2}~.
 \end{equation}
 We want to show that 
 \begin{align}
 &M'(p) \stackrel{?}{\leq} \frac{N(p)}{(1-2p)^2}\\
&\Leftrightarrow \frac{d}{dp}\left(\frac{p(1-p)}{(1-2p)}\log\left(\frac{1-p}{p}\right)\right) \stackrel{?}{\leq} \frac{ \frac{p(1-p)(1-2p)\left(\log\left(\frac{1-p}{p}\right)\right)^2 }{h(p)}}{(1-2p)^2}\\
&\Leftrightarrow \frac{-1}{(1-2p)} +  \log\left(\frac{1-p}{p}\right) + \frac{2p(1-p)}{(1-2p)^2}\log\left(\frac{1-p}{p}\right)  \stackrel{?}{\leq} \frac{p(1-p)\left(\log\left(\frac{1-p}{p}\right)\right)^2 }{h(p)(1-2p)}\\
&\Leftrightarrow -(1-2p) + (1-2p)^2\log\left(\frac{1-p}{p}\right) + 2p(1-p)\log\left(\frac{1-p}{p}\right) \stackrel{?}{\leq}  \frac{p(1-p)(1-2p)\left(\log\left(\frac{1-p}{p}\right)\right)^2 }{h(p)}\\
&\Leftrightarrow h(p)\left(-(1-2p) + (2p^2 - 2p +1)\log\left(\frac{1-p}{p}\right)\right) \stackrel{?}{\leq} p(1-p)(1-2p)\left(\log\left(\frac{1-p}{p}\right)\right)^2~. 
 \end{align}
 Now we expand $h(p) = -p\log(p) - (1-p)\log(1-p)$ and $\log\left(\frac{1-p}{p}\right) = \log(1-p) - \log(p)$ and evaluate both sides of this inequality while collecting the coefficients of $(\log p)^2$, $(\log(1-p))^2$, $\log(p)\log(1-p)$, $\log(p)$ and $\log(1-p)$. After cancellation, we get that we need to show  \begin{equation}\label{final}
 p^2(\log(p))^2 - (1-p)^2(\log(1-p))^2 + (1-2p) \left( \log(p)\log(1-p) + p\log(p)+(1-p)\log(1-p)\right) \stackrel{?}{\leq} 0~.
 \end{equation}
 Define 
 \begin{equation}
F(p) :=  p^2(\log(p))^2 - (1-p)^2(\log(1-p))^2 + (1-2p) \left( \log(p)\log(1-p) + p\log(p)+(1-p)\log(1-p)\right)~. 
 \end{equation}
 \begin{figure}[h!]
  
  \centering
    \includegraphics[scale = 0.75]{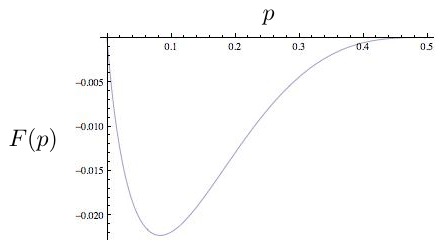}
    \caption{Plot of $F(p)$}
\end{figure}

  \begin{claim}\label{Fp}
 $F(p) \leq 0~.$
 \end{claim}
 \begin{proof}[Proof of Claim \ref{Fp}]
 
 We'll compute the first few derivatives of $F$, and their values at $p = 0$ and $p = \frac{1}{2}$. We'll use $F^{(n)}(p)$ to indicate the $n$-th derivative.
 \begin{align}\label{d1}
 F^{(1)}(p) = \frac{1}{(1-p) p}&(2 (1-p)^2 p \log(1 - p)^2 - 
  p^2 \log(p)(1 - 2 p - 2 (1-p) \log(p)) \\
  &+(1-p) \log(1 - p) (1 - 3 p + 2 p^2 - 2 p \log(p)))~. \nonumber
 \end{align}
 We observe that $F^{(1)}(0) = -1$ and $F^{(1)}(0.5) = 0$. (Note that $F^{(1)}(0)$ is computed in the limit.)
 Now consider the second derivative
 \begin{align}
  F^{(2)}(p) =  \frac{1}{(1-p)^2 p^2}&((-1 + p^2 + 2 p^3 - 2 p^4) \log(1 - p) - 
  2 (1-p)^2 p^2 \log(1 - p)^2\\
   &+ p (-1 + 3 p - 2 p^2 + p (5 - 6 p + 2 p^2) \log(p) + 
     2 (1-p)^2 p \log(p)^2))~. \nonumber
 \end{align}
 Again, evaluating in the limit we see $F^{(2)}(0) \to +\infty$ and $F^{(2)}(0.5) = 0$. Now we compute the third derivative
 \begin{align}
 F^{(3)}(p) = \frac{1}{(1-p)^3 p^3}
 &2 ((1-p)^2 (1 - p^2 + 2 p^3) \log(1 - p)\\
  &+  p (1 + p - 4 p^2 + 2 p^3 + p (2 - 4 p + 5 p^2 - 2 p^3) \log(p)))~. \nonumber
 \end{align}
\begin{figure}[h!]
    \centering
    \includegraphics[scale = 0.75]{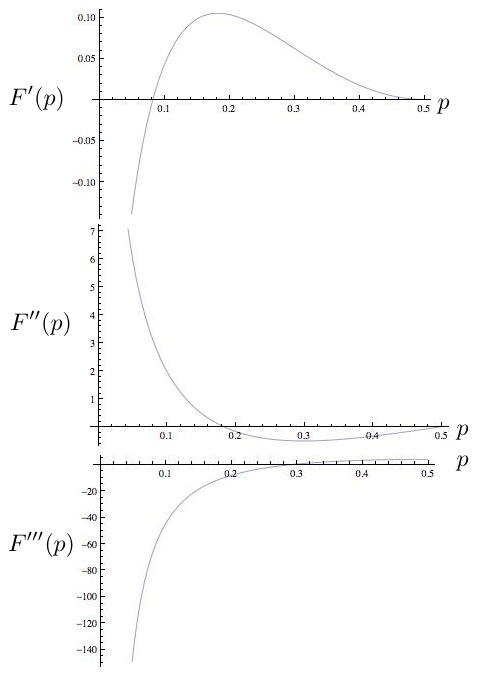}
    \caption{Plot of the derivatives $F(p)$}

\end{figure}

 We evaluate and check that in the limit $F^{(3)}(0) \to -\infty$ and $F^{(3)}(0.5) > 0$.
 Now suppose for some $0 < p < \frac{1}{2}$, it were to be the case that $F(p) > 0$. Since $F(0) = 0$, $F(0.5)=0$ and $F^{(1)}(0) = -1$ we see that $F$ must have a zero in $(0,0.5)$. Now applying Rolle's theorem \cite{apostol1966} twice, we get that $F^{(1)}$ must have $2$ zeros in $(0,0.5)$. We also have $F^{(1)}(0.5) = 0$, which means we can use Rolle's theorem again to conclude that $F^{(2)}$ must have atleast $2$ zeros in $(0,0.5)$. Using $F^{(2)}(0.5) = 0$, and using Rolle's theorem again, we get that $F^{(3)}$ must have atleast $2$ zeros in $(0,0.5)$. Thus, if we can show that $F^{(3)}$ has exactly $1$ zero in $(0,0.5)$, (note that it has atleast $1$ zero since $F^{(3)}(0)$ and $F^{(3)}(0.5)$ have opposite signs) then it implies that $F \leq 0$. Our strategy is to prove $F^{(3)}$ is concave, and based on the values it takes at $0$ and $0.5$, it must have exactly $1$ root in $(0,0.5)$.\\
 To this end, we compute the fifth derivate of $F$  
 \begin{equation}\label{f5}
 F^{(5)}(p) = \frac{2}{(1-p)^5p^5}\left(P_1(p)\log(p) + P_2(p)\log(1-p) + P_3(p)\right)
 \end{equation}
 where
 \begin{eqnarray}
 P_1(p) &=& 2 p^2 (2 - 10 p + 20 p^2 - 11 p^3 + 7 p^4 - 2 p^5)~,\\
 P_2(p) &=& 2 (1-p)^2 (6 - 15 p + 9 p^2 + 3 p^3 - 3 p^4 + 2 p^5)~,\\
 P_3(p) &=& p (12 - 49 p + 70 p^2 - 25 p^3 - 12 p^4 + 4 p^5)~.
 \end{eqnarray}

  \begin{claim}\label{p1p2}
  $P_1(p) \geq 0$, $P_2(p) \geq 0$ for $0 \leq p \leq \frac{1}{2}$.
  \end{claim}
  \begin{figure}[h!]
    \centering
    \includegraphics[scale = 0.75]{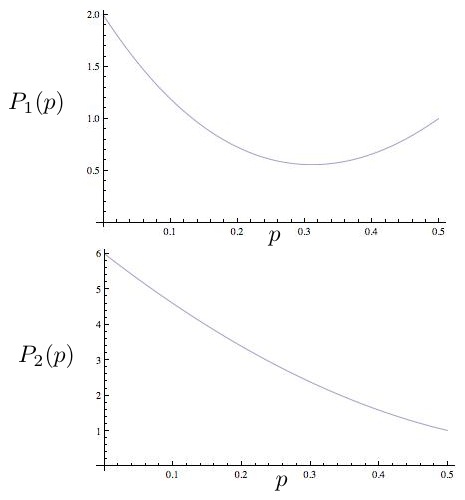}
    \caption{Plots of the $P_1$ and $P_2$}

\end{figure}

  \smallskip

 Assuming Claim \ref{p1p2} is true, we use the following polynomial approximations for $\log(p)$ and $\log(1-p)$:
 \begin{eqnarray}
\log(p) &\leq& -(1-p) - \frac{(1-p)^2}{2}~,\\
\log(1-p) &\leq& -p - \frac{p^2}{2}~.
 \end{eqnarray}
 Thus, 
 \begin{equation}\label{poly}
P_1(p)\log(p) + P_2(p)\log(1-p) + P_3(p)\leq P_1(p)\left( -(1-p) - \frac{(1-p)^2}{2}\right) + P_2(p)\left(-p - \frac{p^2}{2}\right) + P_3(p)~.
 \end{equation}
 The right hand side of the above expression simplifies to to 
 \begin{equation}
 -12(1-p)^2 (p-0.5)^2  p^2 (p^2-p+\frac{7}{3})
 \end{equation}
 which is immediately seen to be $\leq 0$. Thus $(\ref{poly})$ and $(\ref{f5})$ give us 
 $$F^{(5)}(p) \leq 0$$ which means $F^{(3)}$ is concave, thus proving Claim \ref{Fp}. 
 \end{proof}
\begin{proof}[Proof of Claim \ref{p1p2}]
To show $P_1 \geq 0$, we need to show that 
 \begin{equation}
\hat{P_1}(p) =2 - 10 p + 20 p^2 - 11 p^3 + 7 p^4 - 2 p^5 \stackrel{?}{\geq} 0~.
 \end{equation}
 $\hat{P_1}(0) = 2$ and $\hat{P_1}(0.5) = 1$. Thus if we show that $\hat{P_1}$ has no real roots in $(0,0.5)$, we'll be done. We show this using Sturm's theorem \cite{jacobson1985}. Using Mathematica, we construct the Sturm sequence which is
 \begin{align*}
&g_0 = +2 - 10 p + 20 p^2 - 11 p^3 + 7 p^4 - 2 p^5~,\\
&g_1 = -10 + 40 p - 33 p^2 + 28 p^3 - 10 p^4~,\\
&g_2 = -(3/5 - (12 p)/5 + (369 p^2)/50 - (12 p^3)/25)~,\\
&g_3 = -(-(2675/16) + (2625 p)/4 - (61325 p^2)/32)~,\\
&g_4 = -(4436544/150430225 - (16965504 p)/150430225)~,\\
&g_5 = -31638033631325/249850977408~.
\end{align*}
Evaluating the above sequence at $0$, we get the signs $(+,-,-,+,-,-)$ which has $3$ sign changes. Evaluating at $\frac{1}{2}$,we get the signs $(+,+,-,+,+.-)$. Since this also has $3$ sign changes, we are assured that $\hat{P_1}$ has no real roots in $[0, \frac{1}{2}]$.
\smallskip

Similarly, we consider $\hat{P_2}(p)$, for which we need to show
\begin{equation}
\hat{P_2}(p) = 6 - 15 p + 9 p^2 + 3 p^3 - 3 p^4 + 2 p^5 \stackrel{?}{\geq} 0~.
\end{equation}
$\hat{P_2}$ takes values $6$ and $1$ at $0$ and $0.5$ respectively. Again, constructing the Sturm sequence for $\hat{P_2}$ we get
\begin{align*}
&g_0 = 6 - 15 p + 9 p^2 + 3 p^3 - 3 p^4 + 2 p^5~,\\
&g_1 = -15 + 18 p + 9 p^2 - 12 p^3 + 10 p^4~,\\
&g_2 = -(51/10 - (273 p)/25 + (297 p^2)/50 + (12 p^3)/25)~,\\
&g_3 = -(45675/32 - (50825 p)/16 + (61325 p^2)/32)~,\\
&g_4 = -(-(2505792/30086045) + (16965504 p)/150430225)~,\\
&g_5 = (31638033631325/249850977408)~.
\end{align*}
Evaluating at $0$ gives the sign sequence $(+,-,-,-,+,+)$, and evaluating at $0.5$ gives the sign sequence $(+,-,-,-,+,+)$. Since they have the same number of sign changes, we conclude that $\hat{P_2}$ has no zeros in $[0,0.5]$. This proves Claim \ref{p1p2}, and completes the proof of Lemma \ref{dfdx}.

\end{proof}

\end{document}